\documentclass[sigconf,authorversion]{acmart}

\AtBeginDocument{%
  \providecommand\BibTeX{{%
    \normalfont B\kern-0.5em{\scshape i\kern-0.25em b}\kern-0.8em\TeX}}}

\usepackage{balance}
\usepackage{multirow}
\usepackage{booktabs}
\usepackage{pifont}

\usepackage{tabularx} 

\usepackage[ruled, vlined]{algorithm2e}

\usepackage{float}
\usepackage{subfig}
\usepackage{caption}
\usepackage{csquotes}
\usepackage{makecell}
\usepackage[export]{adjustbox}
\newtheorem{problem}{Problem}
\newcommand{\algname}{\textsc{odeN}} 

\copyrightyear{2021}
\acmYear{2021}
\setcopyright{acmlicensed}\acmConference[CIKM '21]{Proceedings of the 30th
    ACM International Conference on Information and Knowledge
    Management}{November  1--5, 2021}{Virtual Event, QLD, Australia}
\acmBooktitle{Proceedings of the 30th ACM International Conference on
    Information and Knowledge Management (CIKM '21), November  1--5, 2021,
    Virtual Event, QLD, Australia}
\acmPrice{15.00}
\acmDOI{10.1145/3459637.3482459}
\acmISBN{978-1-4503-8446-9/21/11}
\begin{document}

\title{\algname: Simultaneous Approximation of Multiple Motif Counts in Large Temporal Networks}


\author{Ilie Sarpe}
\affiliation{%
    \department{Department of Information Engineering}
    \institution{University of Padova}
    \city{Padova}
    \country{Italy}
}
\email{sarpeilie@dei.unipd.it}

\author{Fabio Vandin}
\affiliation{%
    \department{Department of Information Engineering}
    \institution{University of Padova}
    \city{Padova}
    \country{Italy}
}
\email{fabio.vandin@unipd.it}


\begin{abstract}
Counting the number of occurrences of small connected subgraphs, called \emph{temporal motifs}, has become a fundamental primitive for the analysis of \emph{temporal networks}, whose edges are annotated with the time of the event they represent. One of the main complications in studying temporal motifs is the large number of motifs that can be built even with a limited number of vertices or edges.
As a consequence, since in many applications motifs are employed for exploratory analyses, the user needs to iteratively select and analyze several motifs that represent different aspects of the network, resulting in an inefficient, time-consuming process. This problem is exacerbated in large networks, where the analysis of even a single motif is computationally demanding. As a solution, in this work we propose and study the problem of \emph{simultaneously} counting the number of occurrences of \emph{multiple} temporal motifs, all corresponding to the same (static) topology (e.g., a triangle). Given that for large temporal networks computing the exact counts is unfeasible, 
we propose \algname, a sampling-based algorithm that provides an accurate approximation of \emph{all} the counts of the motifs. We provide analytical bounds on the number of samples required by \algname\ to compute rigorous, probabilistic, relative approximations. Our extensive experimental evaluation shows that \algname\ enables the approximation of the counts of motifs in temporal networks in a fraction of the time needed by state-of-the-art methods, and that it also reports more accurate approximations than such methods.
\end{abstract}




\maketitle

\section{Introduction}\label{sec:intro}

Networks are ubiquitous representations that model a wide range of real-world systems, such as social networks \cite{Cho2011}, citation networks \cite{Ding2011Citation}, biological systems~\cite{Girvan2002Networks}, and many others~\cite{Newman2010Networks}. One of the most fundamental primitives in network analysis is the mining of \emph{motifs} \cite{ShenOrr2002, Milo2002, Milo2004Superfamilies} (or \emph{graphlets} \cite{Przulj2007, Bressan2019}), which requires to count the occurrences of small connected subgraphs of $k$ nodes. Motifs represent key building blocks of networks, and they provide useful insights in wide range of applications such as network classification \cite{Shervashidze2009, Milenkovic2008}, network clustering \cite{Baumes2005}, and community detection \cite{Batagelj2003}. 

Modern networks contain rich information about their edges or nodes~\cite{Zong2015, Rossi2021, Kosyfaki2018, Ceccarello2017} in addition to graph structure. One of the most important information is the time at which the interactions, represented by edges, occur. Networks for which such information is available are called \emph{temporal} \cite{Holme2012Temporal,Holme2019Temporal}; novel insights about the underlying \emph{dynamics} of the systems can be uncovered by the analysis of such networks \cite{Kumar2006, Kovanen2013Gender, Kumar2018}. In recent years, many primitives~\cite{Kovanen2011, Hulovatyy2015, paranjape2017motifs, Schwarze2020} have been proposed as counterpart, in temporal networks, to the study of subgraph patterns for nontemporal networks, with each primitive capturing different temporal aspects of a network. One of the most important such primitives is the study of \emph{temporal motifs} \cite{paranjape2017motifs}. Temporal motifs are small connected subgraphs with $k$ nodes and $\ell$ edges occurring with a prescribed order within a time interval of duration $\delta$. Temporal motifs describe the patterns shaping interactions over the network, e.g., networks from similar domains tend to have similar temporal motif counts \cite{paranjape2017motifs}, and their analysis is useful in many applications, e.g., anomalies detection \cite{Belth2020Persistent}, network classification \cite{Tu2019Classification}, and social networks~\cite{boekhout2019efficiently}. 

The temporal dimension poses several challenges in the analyses of motifs. A major challenge is represented by the large number of temporal motifs that can be build even with a limited number of vertices and edges. For example, even considering directed (and connected) temporal motifs with only 3 vertices and 3 edges, there are 32 such motifs. In several domains
when motifs are studied in the exploratory analysis of a temporal network it is almost impossible for the data analyst to known \emph{a priori} which motif is the most interesting and useful.
In social networks, a set of 3 vertices 
represents the smallest non trivial community, and different temporal motifs with 3 vertices describe different patterns of interactions in such community.
Hence, studying all such motifs can provide novel insights on the interactions within such communities.
In network classification, considering the counts of all the 32 motifs with 3 vertices and 3 edges lead to models with improved accuracy~ \cite{Tu2019Classification}.

However, since state-of-the-art approaches for general temporal motifs only allow the analysis of one motif at the time, the user needs to \emph{iteratively} select and analyze the various motifs, resulting in an inefficient and time consuming process, in particular for large networks.

In this paper, we define and study the problem of simultaneously counting the occurrences of various temporal motifs. In particular, we consider all motifs corresponding to the same \emph{static} target template (e.g., all triangles - see Fig.~\ref{subfig:highlevelexp}). This problem is extremely challenging, since computing the count of even a single temporal motif is NP-Hard in general \cite{liu2019sampling}, with existing state-of-the-art approaches having complexity exponential in the number of edges of the motif to obtain even a single motif's count \cite{liu2019sampling, Wang2020Efficient, Sarpe2021PRESTO}.  

The task of counting temporal motifs is hindered by the sheer size of modern datasets and, therefore, scalable techniques are needed to deal with such amount of data. Since exact approaches~\cite{paranjape2017motifs, Mackey2018Temporal, Gurukar2015} are impractical,  rigorous and efficient approximation algorithms providing tight guarantees are needed. In this work we develop \algname, a sampling algorithm that provides a high quality approximation for the problem of counting multiple temporal motifs with the same static topology. 
Our main contributions are as follows:
\begin{itemize}
    \item We propose the \emph{motif template counting problem}, where, given a temporal network, a $k$-node target template graph $H$, the number $\ell$ of edges of each temporal motif, and a bound $\delta$ on the duration of the temporal motifs, the problem requires to output all the counts of the temporal motifs whose static topology corresponds to $H$ and having exactly $\ell$ temporal edges, occurring within $\delta$-time.
    \item We propose \algname, a randomized sampling algorithm providing a high quality approximation for the motif template counting problem. \algname’s approach is to sample a set of motif occurrences, ensuring that they all share the same static topology $H$.  Thus, \algname\ takes advantage of the constraint that all motifs must share a common target template $H$, aggregating the computation of all motif counts in a sample. 
    \algname's approximation, as in other data mining applications, is controlled by two parameters $\varepsilon, \eta$, which control respectively the quality and the confidence of the approximations.
    \item We show a tight and efficiently computable bound on the number of samples required by \algname\ for the approximation to be within $\varepsilon$ error with confidence $>1-\eta$ for \emph{all} temporal motif’s counts.
    \item We perform large scale experiments using datasets with up to billions of temporal edges, showing that \algname\ requires a fraction of the time required by state-of-the-art approximation algorithms for single motif counts, and that it reports sharper estimates.
    We then provide a parallel implementation of \algname\ displaying almost linear speedup in many configurations. We also show how \algname\ provides novel insights on the dynamics of a real-world temporal network. 
\end{itemize}

\section{Preliminaries}\label{sec:prelims}
In this section we introduce the basic notions that we will use throughout the work, and we define the computational problem of counting multiple temporal motifs sharing a common target template graph.
We start by defining temporal networks.

\begin{definition}
    \label{defn:temporal_graph}
    A \emph{temporal network} is a pair $T=(V,E)$ where, $V=\{v_1, \dots , v_n\}$ and $E=\{(x,y,t):x,y \in V, x \neq y, t \in \mathbb{R^{+}}\}$ with $|V|=n$ and $|E|=m$.
\end{definition}
Given $(x,y, t) \in E$, we say that $t$ is the \emph{timestamp} of the directed edge $(x,y)$.
Given a temporal network $T$, by ignoring the timestamps of its edges we obtain the associated \emph{undirected projected static network}, defined as follows.
\begin{definition}
    \label{defn:projected_graph}
    The \emph{undirected 
        projected static network} of a temporal network $T=(V,E)$ is the pair $G_T=(V,E_T)$ 
    that is an undirected 
    network, such that $E_T=\{ \{x,y\} : (x,y,t) \in E \}$.
\end{definition}

We will often use the term \emph{static network} to denote a network whose edges are without timestamps.
Next we introduce the definition of \emph{temporal motifs} as defined by Paranjape et al. \cite{paranjape2017motifs}, which are small, connected subgraphs representing patterns of interest.

\begin{definition}\label{defn:temporal-motif}
    A \emph{$k$-node $\ell$-edge temporal motif} $M$ is a pair $M = (\mathcal{K}, \sigma)$ where $\mathcal{K}=(V_\mathcal{K}, E_\mathcal{K})$ is a directed and weakly connected \emph{multigraph} where $V_{\mathcal K} =\{v_1, \dots , v_k\}$, $E_\mathcal{K}=\{(x,y):x,y \in V_\mathcal{K}, x \neq y\}$ s.t.\  $|V_\mathcal{\mathcal{K}}|=k$ and $|E_\mathcal{K}|=\ell$, and $\sigma$ is an ordering of $E_\mathcal{K}$.
\end{definition} 

\begin{figure}[t]
    \centering
    \subfloat[]{
        \begin{tabular}{lc}
            \includegraphics[width=.28\linewidth]{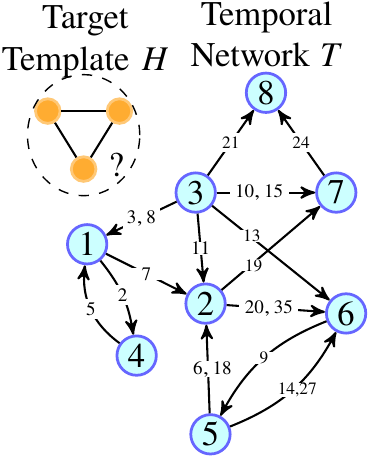} &  \includegraphics[width=.48\linewidth]{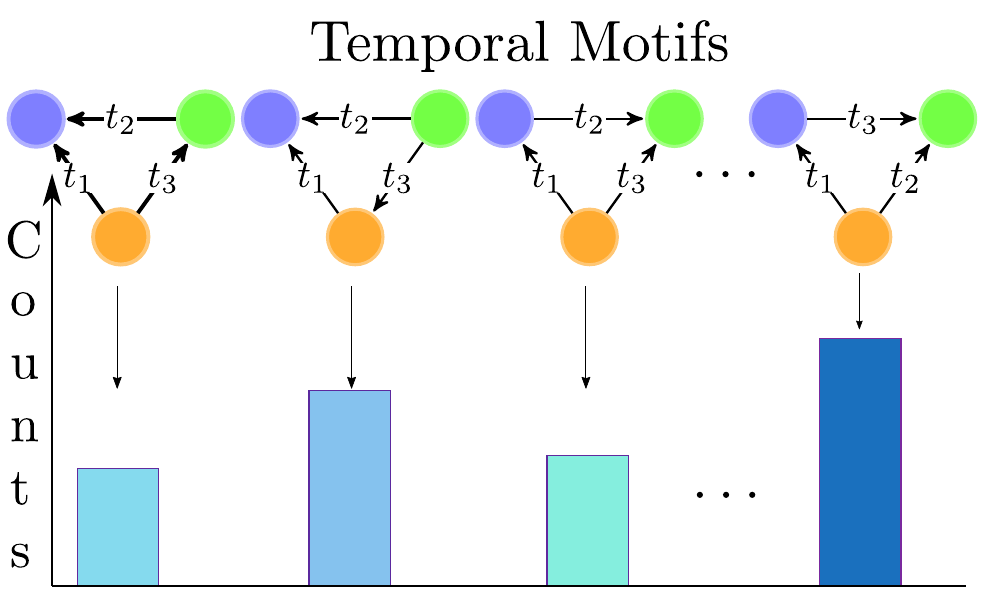} 
        \end{tabular}
        \label{subfig:highlevelexp}
    }\\
    \subfloat[]{
        \includegraphics[width=.3\linewidth]{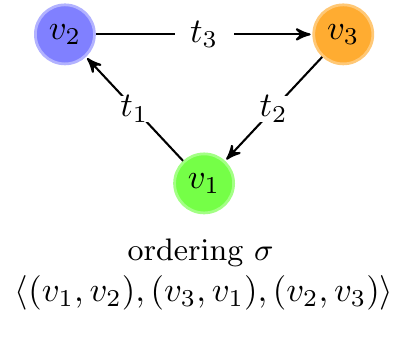}
        \label{subfig:tmotif}
    }
    \subfloat[]{
        \includegraphics[width=.38\linewidth]{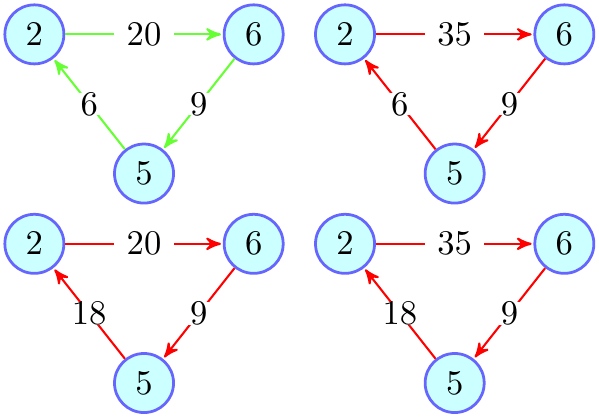}
        \label{subfig:instances}
    }
    \caption{(\ref{subfig:highlevelexp}): Motif template counting problem overview: given a temporal network and a (static) target template, compute the counts of all temporal motifs that map on the template. (\ref{subfig:tmotif}): Temporal motif, with $k=3, \ell=3$, and its ordering $\sigma$. (\ref{subfig:instances}): Sequences of edges of the network in (\ref{subfig:highlevelexp}) among nodes $\{2,5,6\}$ that map topologically on the motif in (\ref{subfig:tmotif}). For $\delta= 15$ only the green sequence is a $\delta$-instance of the motif, since the timestamps respect $\sigma$ and $t_{\ell}'-t_1' = 20-6 \le \delta$. The red sequences are not $\delta$-instances, since they do not respect such constraint or do not respect the ordering $\sigma$.}
    \label{fig:basicdef}
\end{figure}

Note that a $k$-node $\ell$-edge temporal motif $M = (\mathcal{K}, \sigma)$ is also identified by the sequence $\langle(x_1,y_1) , \dots, (x_{\ell},y_{\ell})\rangle$ of edges ordered according to $\sigma$; we will often use such representation for a motif $M$ (see Fig.\ (\ref{subfig:tmotif}) for an example). Given a $k$-node $\ell$-edge temporal motif $M$, the values of $k$ and $\ell$ are determined by $V_\mathcal{K}$ and $E_\mathcal{K}$. We will therefore use the term \emph{temporal motif}, or simply \emph{motif}, when $k$ and $\ell$ are clear from context. Given a temporal motif $M= ((V_\mathcal{K}, E_\mathcal{K}), \sigma)$, we  denote with $G_{u}[M]$ the undirected graph corresponding to the underlying undirected graph structure of the multigraph $\mathcal{K}$ of $M$, that is  $G_{u}[M] = (V_\mathcal{K}, E_M^{u}) $ where $E_M^{u} = \{\{x,y\} : (x,y)\lor (y,x)\in E_\mathcal{K} \}$ (i.e., $E_M^{u}$ is the \emph{set} of undirected edges associated to the multiset $E_\mathcal{K}$). Notice that directed edges of the form $(x,y),(y,x)$ as well as multiple directed edges $(x,y),(x,y),\dots$ from $E_{\mathcal{K}}$ are represented by the same undirected edge $\{x,y\}$ in $E_M^u$. 

For a fixed temporal motif $M$, we are interested in identifying its realizations in $T$ appearing within at most $\delta$-time \emph{duration}, as captured by the following definition.

\begin{definition}\label{defn:delta-instance}
    Given a temporal network $T=(V,E)$ and $\delta \in \mathbb{R^{+}}$, a time ordered sequence $S=\langle(x'_{1}, y'_{1}, t'_{1}), \dots, (x'_{\ell}, y'_{\ell}, t'_{\ell})\rangle$ of $\ell$ unique temporal edges from  $T$ is a $\delta$\textit{-instance} of the temporal motif $M=\langle(x_1,y_1),\dots,(x_{\ell},y_{\ell})\rangle$  if:
    \begin{enumerate}
        \item there exists a bijection $f$ on the vertices such that $f(x'_{i}) = x_i$ and $f(y'_{i}) = y_i, \, i=1, \dots , \ell$; and
        \item the edges of $S$ occur within $\delta$ time, i.e., $t'_{\ell} - t'_{1} \leq \delta$.
    \end{enumerate} 
\end{definition}

Exploring different values of $\delta$ in the above definition often leads to different insights on the temporal network that may be discovered through the analysis of the motifs~\cite{Holme2012Temporal,Panzarasa2009Patterns, Kovanen2011, Bajardi2011}. Note that in a $\delta$-instance of the temporal motif $M=(\mathcal{K},\sigma)$ the edge timestamps must be sorted according to the ordering $\sigma$ (see Fig.\ (\ref{subfig:instances}) for an example). In fact, $\sigma$ plays a key role in defining a temporal motif, with different orderings of the same multigraph $\mathcal{K}$ reflecting diverse dynamic properties captured by the motif.

For a given directed multigraph $\mathcal{K}$ with $|E_{\mathcal{K}}|=\ell$ edges, in general not all the $\ell!$ orderings of its edges define \emph{distinct} temporal motifs. We therefore introduce the following equivalence relation.
\begin{definition}\label{defn:motifs-isorphism}
    Let $M_1$ and $M_2$ be two temporal motifs. 
    Let $M_1  = \langle(x_1^1,y_1^1), \dots,  (x_{\ell}^1,y_{\ell}^1)\rangle$, and  $M_2=\langle(x_1^2,y_1^2),$ $ \dots,  (x_{\ell}^2,y_{\ell}^2)\rangle$ be the sequences of edges of $M_1$ and $M_2$, respectively. We say that $M_1$ and $M_2$ are \emph{not distinct} (denoted with $M_1 \cong_\tau M_2$) if there exists a bijection $g$ on the vertices such that $g(x_{i}^1) = x_i^2$ and $g(y_{i}^1) = y_i^2, \, i=1, \dots , \ell$.
\end{definition}

We  provide an example of the definition above in Figure \ref{fig:congruentMotifs}.

Given two networks (undirected or temporal) $G,G'$ we say that $G'=(V',E')$ is a \emph{subgraph} of $G=(V,E)$ (denoted with $G'\subseteq G$) if $V'\subseteq V $ and $E'\subseteq E$. 
Note that we require a subgraph to be \emph{edge} induced.
To conclude the preliminary notions, we recall the definition of \emph{static graph isomorphism}. 
\begin{definition}\label{defn:imorph_def}
    Given two graphs $G=(V_G,E_G)$ and $H=(V_H,E_H)$ we say that the two graphs are \emph{isomorphic}, denoted with $G\simeq H$ if and only if there exists a bijection $f: V_G \mapsto V_H$ on the vertices such that $ e=(u,v)\in E_G \Leftrightarrow e'=(f(u), f(v)) \in E_H$.
\end{definition}

Let $\mathcal{U}(M,\delta)=\{I: I$ is a $\delta$-instance of $M \}$ be the \textit{set of (all) $\delta$-instances} of the motif $M$ in $T$. The \emph{count} of $M$ is $C_M(\delta) = |\mathcal{U}(M,\delta)|$, denoted with $C_M$ when $\delta$ is clear from the context.

Given a static undirected graph $H$, which we call the \emph{target template}, we are interested in solving the problem of computing the number of $\delta$-instances of all temporal motifs with $\ell$ edges and all corresponding to the same static graph $H$. More formally, given the \emph{target template} $H = (V_{H},E_{H})$, which is a simple and connected graph, and $\ell \geq |E_H| \in \mathbb{Z}_{+}$, let $\mathcal{M}(H,\ell)$ be the set of \emph{distinct} temporal motifs with $\ell$ edges whose underlying undirected graph structure corresponds to $H$, 
that is $\mathcal{M}(H,\ell)$ contains motifs $M_i = ((V_\mathcal{K}^i, E_\mathcal{K}^i), \sigma_i)$, $i=1,2,\dots$, such that i) $G_{u}[M_i] \simeq H$; ii) 
$|E_\mathcal{K}^i| = \ell$; and iii) $M_i \ncong_\tau M_j, \forall j\neq i$.

Let us explain intuitively the constrains above. First, $H$ imposes a constraint on the \emph{undirected} static topology the temporal motifs of interest (that are directed subgraphs) should have. That is, it requires all the motifs to have the same underlying graph structure ($G_{u}[M]$), which must be isomorphic to $H$. This is a useful way to represent multiple related temporal motifs. For example, in social network analysis by fixing $H$ as an undirected triangle we consider in $\mathcal{M}(H,\ell)$ all temporal motifs that characterize the communication between groups of three friends (i.e., each motif will represent a \emph{different} form of communication among all such groups~\cite{paranjape2017motifs}).
The second constraint requires each motif $M_i \in \mathcal{M}(H,\ell)$ to have \emph{exactly} $\ell \ge |E_H|$ edges, with $\ell$ provided in input by the user. Fixing the parameter $\ell$ is motivated by the fact that motifs with different values of $\ell$ (even with the same target template structure $H$) reflect different patterns of interaction (e.g, a group of friends that exchanges $\ell=3$ or $\ell=4$ messages). As we will show empirically in Section~\ref{subsec:caseStudy}, such counts vary significantly with $\ell$ for fixed $H$ and $\delta$. Finally, the third constraint ensures that we only count distinct motifs, i.e., motifs representing different patterns. 

We now define the motif template counting problem.

\begin{problem} \label{problem:TGMC}
    \textbf{Motif template counting problem.} Given a temporal network $T$, a static \emph{undirected} target graph $H = (V_{H}, E_{H})$, $\ell \in \mathbb{Z}_{+},  \ell \geq |E_H|$, and a parameter $\delta \in \mathbb{R}_{+}$, find the counts $C_{M_i}(\delta)$ of motifs $M_i \in  \mathcal{M}(H,\ell), i=1,\dots, |\mathcal{M}(H,\ell)|$ in $T$.
\end{problem}


\begin{figure}[t]
    \centering
    \begin{tabular}{cc}
        \includegraphics[width=.4\linewidth]{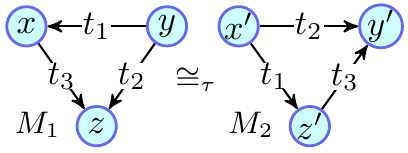} &
        \includegraphics[width = .4\linewidth]{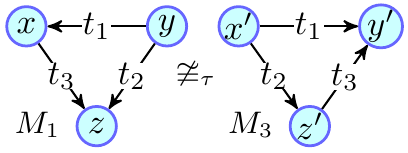}
    \end{tabular}
    \caption{(Left): The two motifs are \emph{not distinct}: let $\sigma_1 = \langle(y,x), (y,z),(x,z)\rangle$ and $\sigma_2 = \langle(x',z'), (x',y'),(z',y')\rangle$ corresponding to $M_1$ and $ M_2$, then the function $f:V_\mathcal{K}^1\mapsto V_\mathcal{K}^2$ defined by $f(x) = z', f(y) = x', f(z) = y'$ preserves both the topology and the ordering as from Definition \ref{defn:motifs-isorphism}. (Right): The two motifs are \emph{distinct} since there is no map $f:V_\mathcal{K}^1\mapsto V_\mathcal{K}^3$ preserving both the topology and ordering.}
    \label{fig:congruentMotifs}
\end{figure}

We now provide an example of the different motifs to be counted for different values of $\ell$ with a fixed target template $H$.
\begin{example}
    Let $H=(\{v_1,v_2\}, \{ \{v_1,v_2\} \})$, that is, the target template is an edge. Let $e_1=(v_1,v_2)$ and $e_2=(v_2,v_1)$. By varying $\ell\in\{2,3\}$ the motifs in $\mathcal{M}(H,\ell)$, for which we want to compute the counts, are: $M_1 = \langle e_1,e_1 \rangle $ and $M_2 = \langle e_1,e_2 \rangle$ for $\ell=2$ (i.e., $|\mathcal{M}(H,2)|=2$) while $M_1 = \langle e_1,e_1,e_1 \rangle, M_2= \langle e_1,e_2,e_1 \rangle , M_3= \langle e_1,e_2,e_2 \rangle  , M_4= \langle e_1,e_1,e_2 \rangle$ for $\ell=3$ (i.e., $|\mathcal{M}(H,3)|=4$).
\end{example}

Since solving the counting problem \emph{exactly} is NP-Hard in general\footnote{The hardness depends on the topology of the motif. For example for triangles and single edges there exist polynomial time-algorithms, even if they are impracticable on very large networks. Interestingly,   counting temporal star-shaped motifs is NP-Hard~\cite{liu2019sampling}, while on static networks such motifs can be counted in polynomial time.} even for one single temporal motif, we aim at providing high-quality approximations to the motif counts as follows.

\begin{problem} \label{problem:Approximation}
    \textbf{Motif template approximation problem.} Given the input parameters of Problem \ref{problem:TGMC} and additional parameters $\varepsilon \in \mathbb{R}_{+}, \eta \in (0,1)$, compute approximations $C_{M_i}'(\delta)$ of counts $C_{M_i}(\delta)$ of motifs $M_i \in \mathcal{M}(H,\ell), i=1,\dots, |\mathcal{M}(H,\ell)|$, such that $\mathbb{P}[\exists i \in\{1,\dots,|\mathcal{M}(H,\ell)| \} :  |C_{M_i}'(\delta) -C_{M_i}(\delta)| \geq \varepsilon C_{M_i}(\delta) ] \leq \eta$, that is $C_{M_i}'(\delta)$ is a relative $\varepsilon$-approximation to the count $C_{M_i}(\delta)$ with probability $\ge 1 - \eta$ for all $i=1,\dots,|\mathcal{M}(H,\ell)|$ simultaneously.
\end{problem}

\section{Related Works}\label{sec:previousWorks}
Much work has been done on enumerating and approximating $k$-node motifs in (nontemporal) networks. We refer the interested reader to the surveys \cite{ribeiro2019survey,yu2020motif}. However, such works cannot be easily adapted to temporal motifs since they do not properly account for the temporal information~\cite{paranjape2017motifs, Holme2012Temporal}.
Many different definitions of temporal networks and temporal patterns have been proposed:  here we will focus only on those works that are relevant for our work, the interested reader may refer to~\cite{Masuda2016Guide, Holme2019Temporal, Holme2012Temporal, Jazayeri2020MotifSurvey} for a more general overview.

Our work builds on the work of Paranjape et al.~\cite{paranjape2017motifs} which first introduced the definition of temporal motif used here, and the problem of counting single temporal motifs. The authors provided a general algorithm for counting a single temporal motif by enumerating all the subsequences of edges that map on a single static subgraph.
Their approach is not feasible on large datasets since it requires exhaustive enumeration of \emph{all} subgraphs of the undirected projected static network  $G_T$ that are isomorphic to the target template $H$. The authors also proposed efficient algorithms and data-structures for counting $3$-node $3$-edge motifs, which may be used for the exact counting subroutines within \algname\ sampling framework. In addition to the algorithmic contributions, the authors also showed that networks from similar domains tend to exhibit similar temporal motif counts. They also showed how motif counts can provide significant insights on the communication patterns in many networks, highlighting the importance of studying temporal motifs in temporal networks.

Other \emph{exact} algorithms have been proposed for the problem of counting a single motif, or for slightly different problems.  
Mackey et al.~\cite{Mackey2018Temporal} presented a backtracking algorithm for counting a single temporal motif that can be use for any motif.
Boekhout et al.~\cite{boekhout2019efficiently} developed exact algorithms for counting temporal motifs in multilayer temporal networks (i.e., each edge is a tuple $(x,y,t,a)$ with $a$ denoting the layer of each edge), they also discuss efficient data-structures for counting 4-node 4-edge motifs, which may also be adapted for the exact counting subroutines in our sampling framework \algname. 
Being exact, both such algorithms do not scale on massive datasets due large  time and memory requirements.


Several approximation algorithms have been proposed in recent years for estimating the count of a \emph{single} motif. Liu et al.~\cite{liu2019sampling} proposed a temporal-partition based sampling approach. 
Wang et al.~\cite{Wang2020Efficient} introduced a sampling-based algorithm that selects temporal edges with a fixed probability specified by the user.
Lastly, Sarpe and Vandin~\cite{Sarpe2021PRESTO} proposed PRESTO, an algorithm based on uniform sampling of small  windows of the temporal network $T$.
All such sampling algorithms can be used to analyze a single temporal motif but become inefficient as the number of motifs to be counted grows, such as in Problem~\ref{problem:Approximation}. In fact, they cannot leverage the additional information that all motifs $M_1,\dots,M_{|\mathcal{M}(H,\ell)|}$ must share a common static topology isomorphic to $H$. As stated in Section~\ref{sec:intro}, when analysing a temporal networks it is hard to know \emph{a-priori} which motif is representing important functions for the network, therefore one often relies on testing all possible orderings $\sigma$ over one fixed target template $H$ for fixed $\ell,\delta$~\cite{paranjape2017motifs, Tu2019Classification} (as in Prob.\ \ref{problem:TGMC}) resulting in a time consuming and inefficient procedure. Our approach instead supports the direct analysis of \emph{multiple} temporal motifs, enabling the study of hundreds of temporal motifs on massive networks in a very limited time.

\section{\algname}
In this section we present \algname, our algorithm to address the motif template approximation problem (Prob.\ \ref{problem:Approximation}). We start in Section~\ref{subsec:Overview} with an overview of \algname. We then describe the algorithm in Section~\ref{subsec:generalAlg}, analyze its time complexity in Section~\ref{subsubsec:computComplex} and its theoretical guarantees, including an efficiently computable bound on the number of samples required to obtain the desired probabilistic guarantees, in Section~\ref{sec:approx_analysis}.

\subsection{Overview of \algname}\label{subsec:Overview}
Our algorithm \algname\ estimates of the counts of motifs in $\mathcal{M}(H,\ell)$. The main idea is to avoid the explicit generation all the motifs $M_i \in \mathcal{M}(H,\ell), i=1,\dots,|\mathcal{M}(H,\ell)|$ to count them one at the time as it is required by existing algorithms that approximate a single motif count.  \algname\ instead leverages the fact that the topology of all motifs must to be isomorphic to the target template $H$, by reusing the computation while estimating the motif counts. 

An overview of the main strategy adopted by our algorithm is presented in Figure \ref{fig:pipeline}. Given the input parameters of Problem~\ref{problem:Approximation}, where $H$ is the target template, the idea behind our procedure is to consider the undirected static projected graph $G_T$ of the input temporal network $T$ and proceed as follows: i) find a set of subgraphs in the static graph $G_T$ that are isomorphic to $H$ 
by first \emph{sampling} an edge $e_R$ of $G_T$ with some probability $p_{e_R}$, where $p_{e_R}$ depends, potentially, on $e_R$ and the temporal network $T$, and then enumerating all subgraphs of $G_T$ isomorphic to $H$ and containing $e_R$; ii) for each such subgraph, consider the corresponding temporal subgraph and 
compute all the counts of the subsequences of  $\ell$ edges occurring within $\delta$-time in such temporal subgraph; iii) for each such subsequence identified, find the corresponding motif in $\mathcal{M}(H,\ell)$, for which the subsequence is a $\delta$-instance of, and update a count for each motif identified; 
iv) weight each motif count opportunely in order to maintain an unbiased estimate of global motif counts;
v) repeat steps i)-iv) a sufficient number of iterations to guarantee the desired $(\varepsilon, \eta)$-approximation (see Problem~\ref{problem:Approximation}). 

\begin{figure}[t]
    \centering
    \includegraphics[width =1\linewidth]{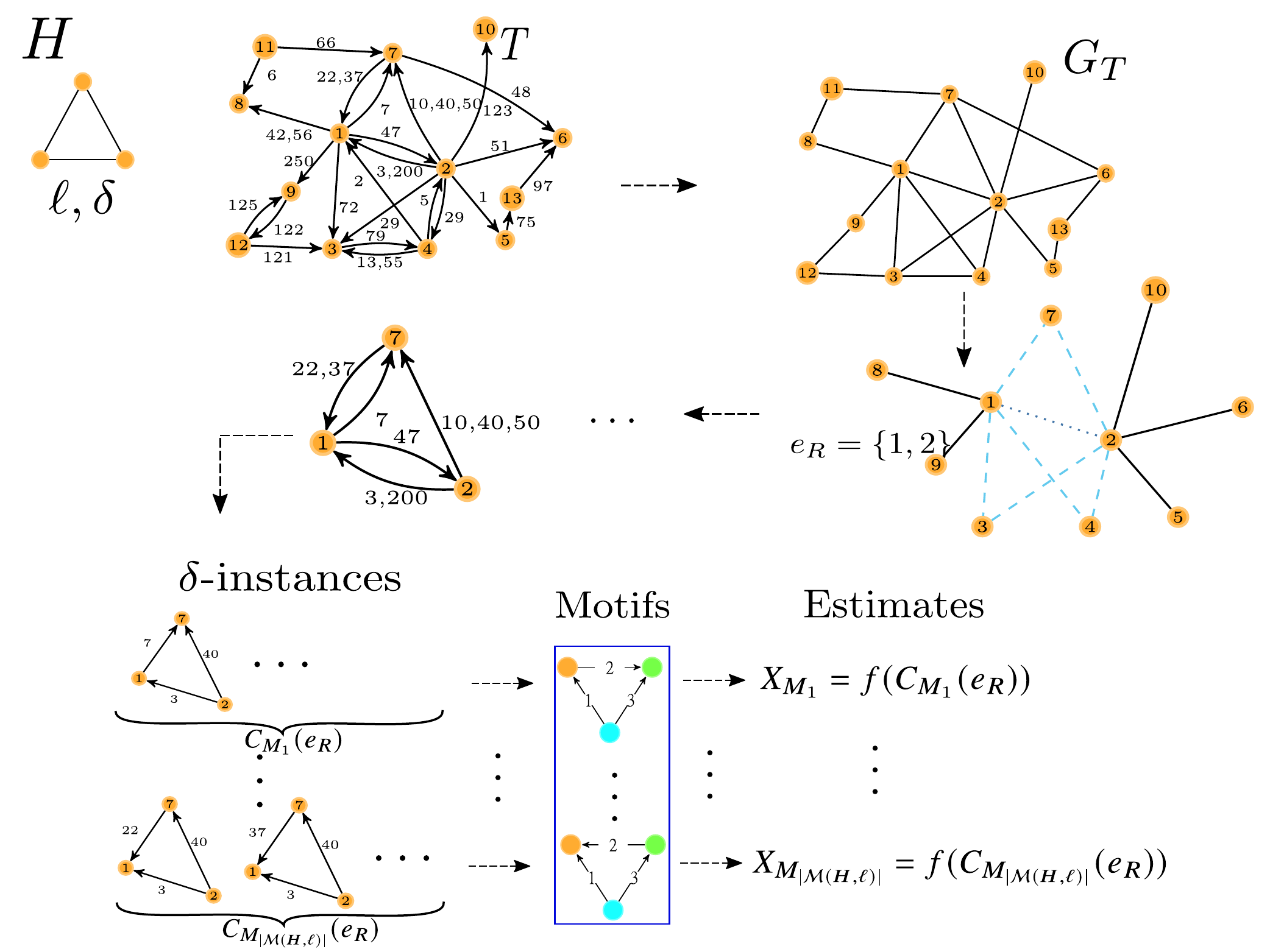}
    
    \caption{Overview of \algname's approximation strategy. Let $H$ be a triangle, and $\ell=3, \delta=40$. \algname\ first collects the static projected network $G_T$, then samples an edge $e_R\in G_T$ randomly ($e_R=\{1,2\}$ in the figure) and enumerates all the subgraphs of $G_T$ isomorphic to $H$ containing $e_R$. For each subgraph it collects the corresponding temporal network, counts the $\delta$-instances of the motifs, and combines the different counts to obtain unbiased estimates of motif counts. This procedure is repeated to obtain concentrated estimates.}
    \label{fig:pipeline}
\end{figure}

\subsection{Algorithm Description} \label{subsec:generalAlg}
\algname\ is described in Algorithm~\ref{alg:samplebystatic}. It first computes $G_T=(V,E_T)$, the undirected projected static graph of $T$ (line~\ref{alglinemain:staticgraph}), and initializes $C_{estimates}$ (line~\ref{alglinemain:mapestimates}) used to store the estimates of motif counts, which are used to compute the estimators $C_{M_i}', i=1,\dots,|\mathcal{M}(H,\ell)|$. Then it repeats $s$ times (line~\ref{alglinemain:forS}) the following procedure: i) pick a random edge $e_R$ from $G_T$ (line \ref{alglinemain:sample edge}) according to some probability distribution over the edges of $E_T$; ii) enumerate all the subgraphs $h$ of $G_T$ such that $h\simeq H$ and $e_R \in h$ (line~\ref{alglinemain:exactsubgriso}); note that this enumeration step is local to $e_R$; iii) for each such $h$ (line~\ref{alglinemain:forStaticSubgraph}), collect the corresponding temporal graph, i.e., all edges in $T$ for which their static projected edge is an edge of $h$ (line~\ref{alglinemain:instancetemporalgraph}), sort the sequence of edges of such graph by increasing timestamps and apply some pruning criteria (lines \ref{alglinemain:sortsequence}-\ref{alglinemain:continueprune});  iv) if the sequence is not pruned, then update the estimates of the number of $\delta$-instances of each temporal motif by calling the routine \texttt{FastUpdate} (line \ref{alglinemain:fastupdatecall}). \texttt{FastUpdate} features an efficient implementation of the general algorithm by Paranjape et al.~\cite{paranjape2017motifs}, for which we devised efficient encodings of the motifs within integers through bitwise operations. 
Such function updates $C_{estimates}$ in order to maintain for each motif the count that will be used to output its unbiased estimate (see Appendix \ref{app:subroutines}).
Let $C_{M_i}(e)$ be the number of $\delta$-instances in $T$ of $M_i, i=1,\dots,|\mathcal{M}(H,\ell)|$ whose undirected projected static network contains edge $e\in G_T$.
\texttt{FastUpdate} updates the estimate of the counts for each motif $M_i$ by summing its unbiased estimate  obtained at the $j$-th iteration (i.e., $X_{M_i}^j=C_{M_i}(e_R)/(|E_H|p_{e_R})$). Once the procedure is repeated $s$ times, for each motif $M_i \in \mathcal{M}(H,\ell), i=1,\dots,|\mathcal{M}(H,\ell)|$, \algname\ computes the final estimate $C_{M_i}'=\frac{1}{s} \sum_{j=1}^s X_{M_i}^j$ where $ X_{M_i}^j = \frac{1}{|E_H|} \sum_{e \in G_T} C_{M_i}(e)X_e/p_e$ is the estimate obtained at the $j$-th iteration (with $X_e$ being a bernoulli random variable denoting if edge $e\in G_T$ is sampled at the $j$-th iteration, s.t. $\mathbb{P}[X_e =1] = p_e$) and outputs it together with the motif (we output $\sigma_i$ over the node-set $V_H$) (lines \ref{alglinemain:fixestimateM}-\ref{alglinemain:output}). We show in Lemma \ref{lemma:unbiased} that \algname\ outputs unbiased estimates for all the motif counts.

We briefly discuss the pruning criteria used in line~\ref{alglinemain:continueprune}. Given a candidate temporal graph $S$ for which $G_S\simeq H$ holds, we check in linear time if $S$ can contain a $\delta$-instance of a motif or not:  since $S$ is already sorted by increasing timestamps (see line \ref{alglinemain:sortsequence}), we efficiently check if there are at least $\ell$ edges within $\delta$-time. If not, then we prune the sequence (since by definition a $\delta$-instance of a motif with $k$-nodes, and $\ell$-edges must have $\ell$ edges occurring within $\delta$-time). We thus avoid calling the subroutine \texttt{FastUpdate}, which has an exponential complexity in general (see Section \ref{subsubsec:computComplex}), on  $S$. 

\begin{algorithm}[t]
    \DontPrintSemicolon
    \SetKwComment{Comment}{$\triangleright$\ }{}
    \LinesNumbered 
    \KwIn{$T=(V,E),H=(V_H,E_H), \delta, s, \ell$}
    \KwOut{ $(M_i, C_{M_i}'), i=1,\dots,{|\mathcal{M}(H,\ell)|}$ where $C_{M_i}'$ is an estimate of $C_{M_i}$ for the motifs in $\mathcal{M}(H,\ell)$.}
    $G_{T} = (V,E_T) \gets \texttt{UndirectedStaticProjection}(T)$\label{alglinemain:staticgraph}\;
    $C_{estimates} \gets \{\}$\label{alglinemain:mapestimates}\;
    \For{$j \gets 1$ \text{to} $s$\label{alglinemain:forS}}
    {
        $e_{R} = \{x_R, y_R\} \gets \texttt{RandomEdge}(p(e): e \in E_T)$ \label{alglinemain:sample edge}\;
        $\mathcal{H} \gets \{h \subseteq G_T : h\simeq H, \{x_R, y_R\} \in h\}$\label{alglinemain:exactsubgriso}\; 
        \ForEach{$h \in \mathcal{H}$ \label{alglinemain:forStaticSubgraph}}
        {
            $S \gets \{(x,y,t),(y,x,t) \in E: \{x, y\}\in h \}$\label{alglinemain:instancetemporalgraph}\;
            \texttt{SortInPlace}($S$) \label{alglinemain:sortsequence}\Comment*[r]{By increasing timestamps}
            \If{\emph{\texttt{*Pruning criteria are \emph{not} met*}} \label{alglinemain:continueprune}}{$\texttt{FastUpdate}(\delta, S, C_{estimates}, p(e_R), H)$ \label{alglinemain:fastupdatecall}\;}

        }
    }
    \ForEach{$(M,X_M) \in C_{estimates}$}
    {
        $C_{M}' \gets \frac{X_{M}}{s}$\label{alglinemain:fixestimateM}\;
        \textbf{output} $(M, C_{M}')$\label{alglinemain:output}\;

    }
    \caption{\algname}\label{alg:samplebystatic}
\end{algorithm}

We now discuss the probability distribution used to sample a random edge $e_R$ from $G_T$ (line \ref{alglinemain:sample edge}), while we describe the subroutine \texttt{FastUpdate} that updates the motif estimates at each iteration (line \ref{alglinemain:fastupdatecall}) and the algorithms employed for the static enumeration in Appendix \ref{app:subroutines} for space constraints (Sections~\ref{subsec:FastUpdate} and \ref{subsubsec:exactRoutines}).

Since our final estimate is an average over $s$ samples of the variables $X_{M_i}^j, i=1,\dots,|\mathcal{M}(H,\ell)|, j=1,\dots,s$, and given that $X_{M_i}^j$ is an unbiased estimate (see Lemma~\ref{lemma:unbiased}) the final estimate is also a consistent estimator (i.e., it converges to $C_{M_i}$ as $s\to \infty$) if each edge has a positive probability of being sampled\footnote{More formally it is only necessary to assign to each $\delta$-instance a known positive sampling probability.}. Thus \emph{any} probability mass assigning positive probabilities on edges can be adopted. We considered different distributions over the edges of $E_T$:
\begin{enumerate}
    \item \emph{Uniform}: $p_e = 1/|E_T|, e \in E_T$;
    \item \emph{Static degree based}: $p_{e} = d(e)/(\sum_{e' \in E_T} d(e')), e \in E_T $ where $d(e=\{x,y\}) = d(x)+d(y)$ is the degree of the edge as sum of the degree of its nodes $x,y \in V$ in $G_T$;
    \sloppy{\item \emph{Temporal degree based}: $p_{e} = \phi(e)/(\sum_{e' \in E_T} \phi(e'))$ with $\phi(e=\{x,y\})=|\{t: \exists(x,z,t) \lor (z,x,t) \in E\}|+|\{t: \exists(z,y,t) \lor (y,z,t) \in E, z \neq x\}|, e\in E_T$;}
    \item \emph{Temporal edge weight based}: $p_{e=\{x,y\}} = |\{(x,y,t), (y,x,t) \in E \}|/m, e\in E_T$;
\end{enumerate}

We empirically found the distribution (4) to be the fastest to converge for small number $s$ of iterations, thus we use it in our ana\-lysis. We observe that many other candidate distributions can be designed (e.g., combining two of those already listed with weights $\xi, 1-\xi, \xi \in (0,1)$) making our framework extremely versatile.

We conclude by summarizing some nice properties of our algorithm:
1) it computes the estimates only for the temporal motifs occurring in the input temporal network $T$ (except for the very unpractical case where the motifs in $\mathcal{M}(H,\ell)$ have all zero counts) 
without generating all the possible candidates, while existing sampling techniques require to first generate all the candidates and then to execute the algorithms on such candidates, even for motifs with zero counts;
2) it takes advantage of the constraint that all motifs share the same underlying topology ($H$), saving computation when estimating the different counts;
3) it is trivially parallelizable: all the $s$ iterations can be executed in parallel;
4) it can easily use most of the fast state-of-the-art subgraph enumeration algorithms developed for the exact subgraph isomorphism problem (see Appendix \ref{subsubsec:exactRoutines}).

\subsection{Time Complexity}\label{subsubsec:computComplex}
In this section we briefly describe the time complexity of \algname. \algname\ needs to compute the probabilities $p(e)$ of edges in advance, which requires a $O(|E_T|)$ preprocessing step. Interestingly, this step does not depend on the target template $H$, so it can be reused for different target templates $H$. One of the most expensive steps in Algorithm~\ref{alg:samplebystatic} is the local enumeration to identify the set $\mathcal{H}$ which in general requires exponential time (line~\ref{alglinemain:exactsubgriso}). For specific topologies this step can be implemented very efficiently with symmetry breaking conditions and min-degree expansion. For example, if $H$ is a triangle this ``local" enumeration to $e_R=\{x_R,y_R\}$ can be done in $O(\min({d_{x_R},d_{y_R}}) )$ time. Let $|\mathcal{H}^*|$ be the maximum cardinality of a set of subgraphs isomorphic to $H$ and adjacent to an edge in $G_T$. Let $|S^*|$ denote the maximum cardinality of a set $S$ collected (in line \ref{alglinemain:instancetemporalgraph}) by our algorithm \algname. Sorting $S^*$ requires $O(|S^*|\log|S^*|)$ time. The subroutine \texttt{FastCount} has a complexity dominated by $O((|S^*|+\ell) |E_H|^{\ell})$ (see~\cite{paranjape2017motifs} and App.~\ref{subsec:FastUpdate} for more details). So overall the complexity of our procedure is
$O(|E_T| + s(\zeta_{enum} + |\mathcal{H}^*|(|S^*|\log(|S^*|) + |E_H|^{\ell}(|S^*|+\ell) )))$, where $\zeta_{enum}$ is the time required by the static enumerator used as subroutine to compute the set $\mathcal{H}^*$. Such step in general is exponential in the number of edges of $|E_T|$ and depends on the exact technique used as subroutine. The final complexity accounts for the cycle (in line \ref{alglinemain:forS}) that is repeated $s$ times. The parallel version of our algorithm, which executes the cycle of line \ref{alglinemain:forS} in parallel on $\omega$ processing units available, leads to a time complexity of $O(|E_T| + s/\omega(\zeta_{enum} + |\mathcal{H}^*|(|S^*|\log(|S^*|) + |E_H|^{\ell}(|S^*|+\ell) )))$. 

\subsection{Theoretical Guarantees} \label{sec:approx_analysis}
In this section we present the theoretical guarantees provided by \algname. 
All proofs are provided in Appendix \ref{app:proofs}.

Recall that our algorithm outputs, for each motif $M_i \in \mathcal{M}(H,\ell), i=1,\dots,|\mathcal{M}(H,\ell)|$, the following estimate: \sloppy{$C_{M_i}' = \frac{1}{s} \sum_{j=1}^s X_{M_i}^j= \frac{1}{s|E_H|}$ $ \sum_{j=1}^s\sum_{e \in G_T} C_{M}(e)X_e/p_e$.} The following shows that such estimates are unbiased estimates of $C_{M_i}, i=1,\dots,|\mathcal{M}(H,\ell)|$.
\begin{lemma}\label{lemma:unbiased}
    For each motif-count pair $(M_i, C_{M_i}')$ reported in output by \algname, $C_{M_i}'$ is an unbiased estimate to $C_{M_i}$, that is $ \mathbb{E}[C_{M_i}'] = C_{M_i}$
\end{lemma}

Let $\alpha = \min_{\{x,y\}\in E_T}\{|\{(x,y,t),(y,x,t)\in E\}|\}$, i.e., the minimum number of temporal edges of $T$ that map on an edge in $G_T$. We now give an upper bound to the variance of the estimates provided by Algorithm \ref{alg:samplebystatic} for each motif reported in output.
\begin{lemma}\label{lemma:variancebound}
    For each motif-count pair $(M_i, C_{M_i}')$ reported in output by \algname, it holds $\text{Var}[C_{M_i}'] \le   \frac{C_{M_i}^2}{s}\left(\frac{m}{\alpha|E_H|} -1\right)$
\end{lemma}

To give a bound on the number $s$ of samples required by \algname\ to output a $\varepsilon$-approximation that holds on all motifs in output with probability $>1-\eta$, we combine Bennett's inequality \cite{Bennett1962}, an advanced result on the concentration of sums for independent random variables as reported in \cite{Sarpe2021PRESTO}, with a union bound, obtaining the following main result.

\begin{theorem}\label{theo:boundsamplesmultiple} 
    Let $s$ be the number of iterations of \algname, let $\varepsilon \in \mathbb{R}^+$, and $\eta \in (0,1)$. If $s \ge \left(\frac{ m}{\alpha|E_H|} -1\right) \frac{1}{(1+\varepsilon)\ln(1+\varepsilon)-\varepsilon} \ln\left(\frac{2 |\mathcal{M}(H,\ell)| }{\eta}\right)$ then
    \begin{equation*}
    \mathbb{P}[\exists i \in \{1,\dots,|\mathcal{M}(H,\ell)|\}: |C_{M_i}' - C_{M_i}| \ge \varepsilon C_{M_i}] \le \eta.
    \end{equation*}
\end{theorem}


\section{Experimental Evaluation}\label{sec:experimentalEval}

We implemented \algname\ and tested it on several large datasets (see Section~\ref{subsec:expsetup} for details on setup, and data). Our experimental evaluation has the following goals: compare \algname\ with state-of-the-art algorithms for approximating motif counts (Section~\ref{subsec:ApproxRes}); 
evaluate the scalability of a simple parallel implementation of \algname\ (Section~\ref{subsec:ParallelImpl}); provide a case study highlighting the usefulness of using \algname\ (Section~\ref{subsec:caseStudy}) to analyze real-world temporal networks.

\subsection{Setup, and Datasets}\label{subsec:expsetup}
We briefly describe the setup and the large-scale datasets used in our experimental evaluation.

We implemented our algorithm \algname\ in C++20 and compiled it under \texttt{gcc 9.3} with optimization flag enabled (implementation 
available at \url{https://github.com/VandinLab/odeN}), additional details on the implementation are in Appendix \ref{app:implementation}. We compared \algname\ with four different baselines, denoted as \texttt{PRESTO-A} (\texttt{PR-A}), \texttt{PRESTO-E} (\texttt{PR-E})~\cite{Sarpe2021PRESTO}, \texttt{LS}~\cite{liu2019sampling}, and \texttt{ES}~\cite{Wang2020Efficient}. We used the original implementations available from the authors.  We performed all experiments under Ubuntu 20.04 on a machine with 64 cores, Intel Xeon E5-2698 2.3GHz, running each algorithm single threaded and with 300GB of maximum RAM allowed. 

The datasets used in our experimental evaluation are reported in Table \ref{tab:datasets}, which shows the number of nodes and edges of $T$, the precision of the timestamps, the timespan of the network, the number  $|E_T|$ of undirected  edges in the corresponding undirected projected static network $G_T$, the maximum degree $d_{\max}$ of a node in $G_T$ and the maximum number  $w_{\max}$ of temporal edges that are mapped on the same static edge in $G_T$. The datasets are from different domains: SO is a network that models interactions from the Stack-Overflow platform \cite{paranjape2017motifs}, BI is a network of Bitcoin transactions~\cite{liu2019sampling}, RE a network built from comments on the platform Reddit \cite{liu2019sampling}, and EC is a \emph{bipartite} temporal network build from IPv4 packets exchanged between Chicago and Seattle \cite{Sarpe2021PRESTO}. See the original papers for more details on the networks and the processes they model. 

When measuring the running times for the various algorithms we exclude the time to read the dataset. Since \texttt{ES}'s implementation supports only values of $\ell$ up to 4, we do not report results for \texttt{ES} and $\ell > 4$.
Unless otherwise stated we used $\delta=86400$ for SO and RE, $\delta=43200$ on BI, and $\delta=50000$ on EC, as done in previous works \cite{paranjape2017motifs,liu2019sampling, Wang2020Efficient}. 
Since all algorithms used in our comparison have different parameters and only \algname\ counts multiple motifs simultaneously, we used the following procedure to choose the parameters. For a given target template $H$ and $\ell$, we run \texttt{PRESTO-A}, \texttt{PRESTO-E}, \texttt{LS}, and \texttt{ES} for each motif in $\mathcal{M}(H,\ell)$ with fixed parameters, and computed their running time as the sum of the running times required by the single motifs in $\mathcal{M}(H,\ell)$. We then fixed the parameters of \algname\ so that its running time would be at most the same as  the other methods, or be close to it. All the parameters used in the experiments (including sample sizes) are reported with the source code.
To extract the \emph{exact} counts of motifs we used a modified version of the algorithm by Mackey et al.~\cite{Mackey2018Temporal}. We do not report the running times of such algorithm since, even though it employs parallelism, it still runs several orders of magnitude slower than approximate approaches.

\begin{table}[t]
    \centering
    \caption{Datasets used and their statistics. See Section \ref{subsec:expsetup} for details on the statistics reported.}
    \label{tab:datasets}
    \scalebox{0.7}{
        \begin{tabular}{cccccccl}
            \toprule
            Name& $n$&$m$& $|E_T|$ & $ d_{\max}$& $w_{\max}$ & Precision & Timespan\\
            \midrule
            SO & 2.58M & 47.9M & 28.1M  & 44K & 594 & sec & 2774 (days)\\
            BI & 48.1M & 113M & 84.3M  & 2.4M & 24.2K &sec & 2585 (days)\\
            RE & 8.40M & 636M & 435.3M  & 0.3M & 165K & sec &3687 (days)\\
            EC & 11.16M & 2.32B & 66.8M  & 0.3M & 3.8M & $\mu$-sec &62.0 (mins)\\
            \bottomrule
        \end{tabular}
    }
\end{table}

\subsection{Approximation Quality and Running Time}\label{subsec:ApproxRes}

In this section we compared the quality of the estimates and the running times of \algname\ and the baseline sampling approaches.

To evaluate the approximations qualities we used the MAPE (Mean Average Percentage Error) metric over ten executions of each algorithm and parameter configuration. The MAPE is computed as follows: let $C_{M_i}'$ be the estimate of $C_{M_i}, i=1,\dots,{|\mathcal{M}(H,\ell)|}$, returned by an algorithm, then the relative error of such estimate is $|C_{M_i}'-C_{M_i}|/ C_{M_i}$. The MAPE is the average over the ten runs of the relative errors, in percentage. On each of the ten runs we also measured the running time of each algorithm, for which we will report the arithmetic mean. 

We first discuss the quality of the estimates for different datasets when $H$ is a triangle and $\ell \in \{4,5\}$. For $\ell=4$ there are ${|\mathcal{M}(H,\ell)|} = 96$ triangles, while for $\ell=5, {|\mathcal{M}(H,\ell)|}$ is 800. So as long as $\ell$ increases the approximation task becomes more challenging, due to the exponential growth of the number of motifs. We also observe that, to the best of our knowledge, such a huge number of temporal motifs was never tested before on large datasets due to the limitations of existing algorithms, while, as we will show, \algname\ renders the approximation task practical even on hundreds of motifs.

\begin{figure}[t]
    \centering
    \subfloat[]{
        \begin{tabular}{cc}
            \includegraphics[width=.45\linewidth]{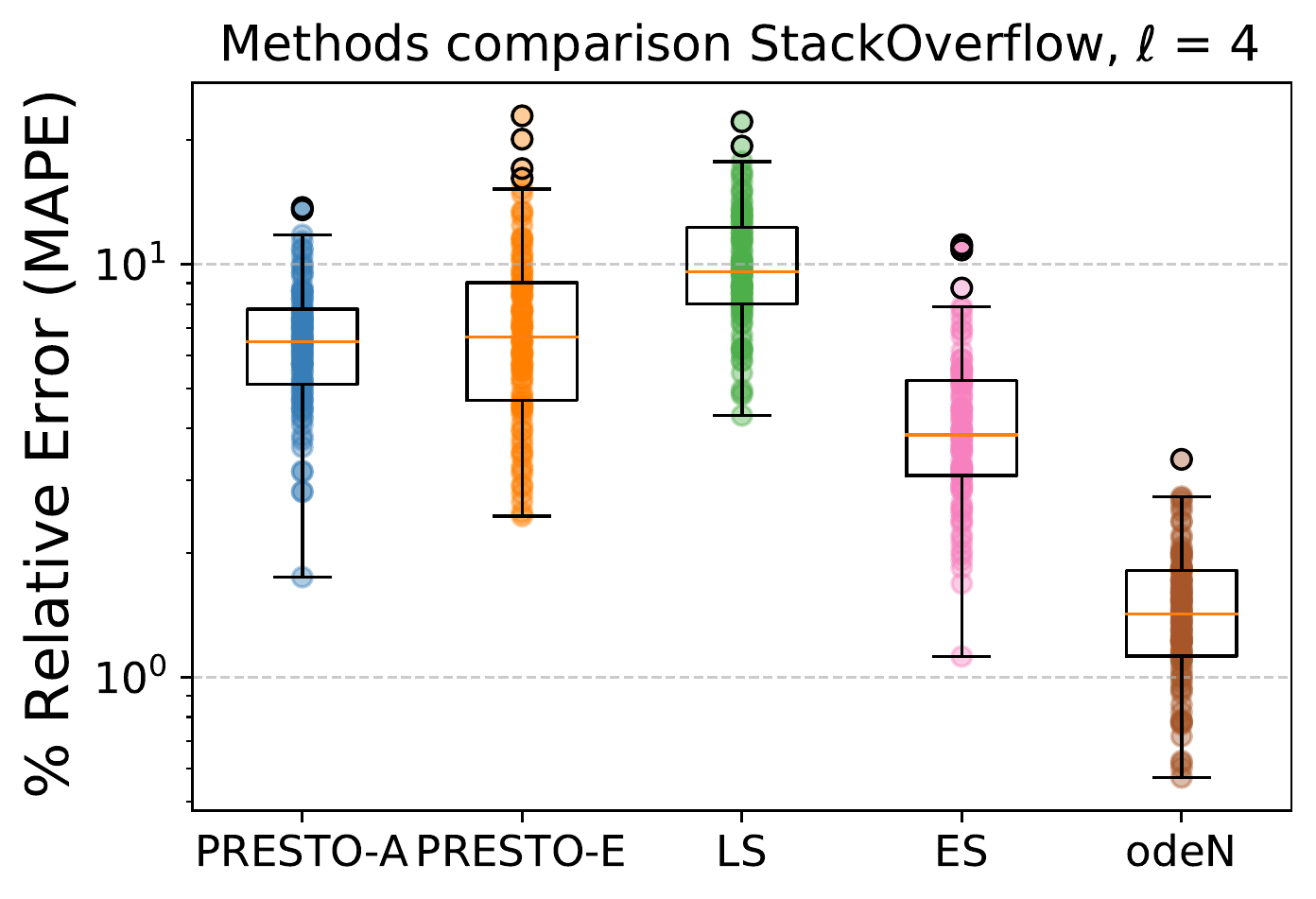} &  \includegraphics[width=.45\linewidth]{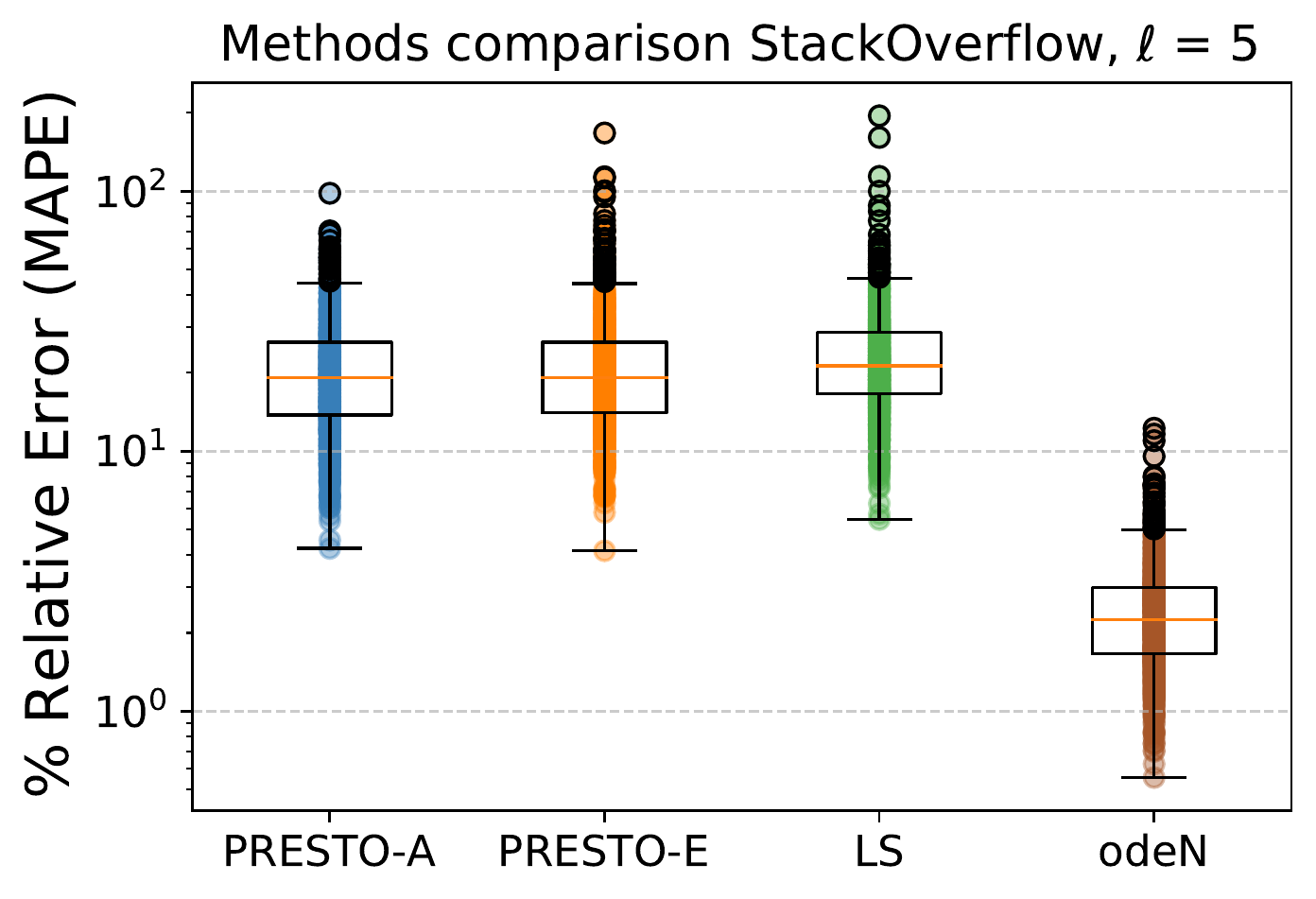} 
        \end{tabular}
        \label{fig:SO_Triangles}
    }\\
    \subfloat[]{
        \begin{tabular}{cc}
            \includegraphics[width=.45\linewidth]{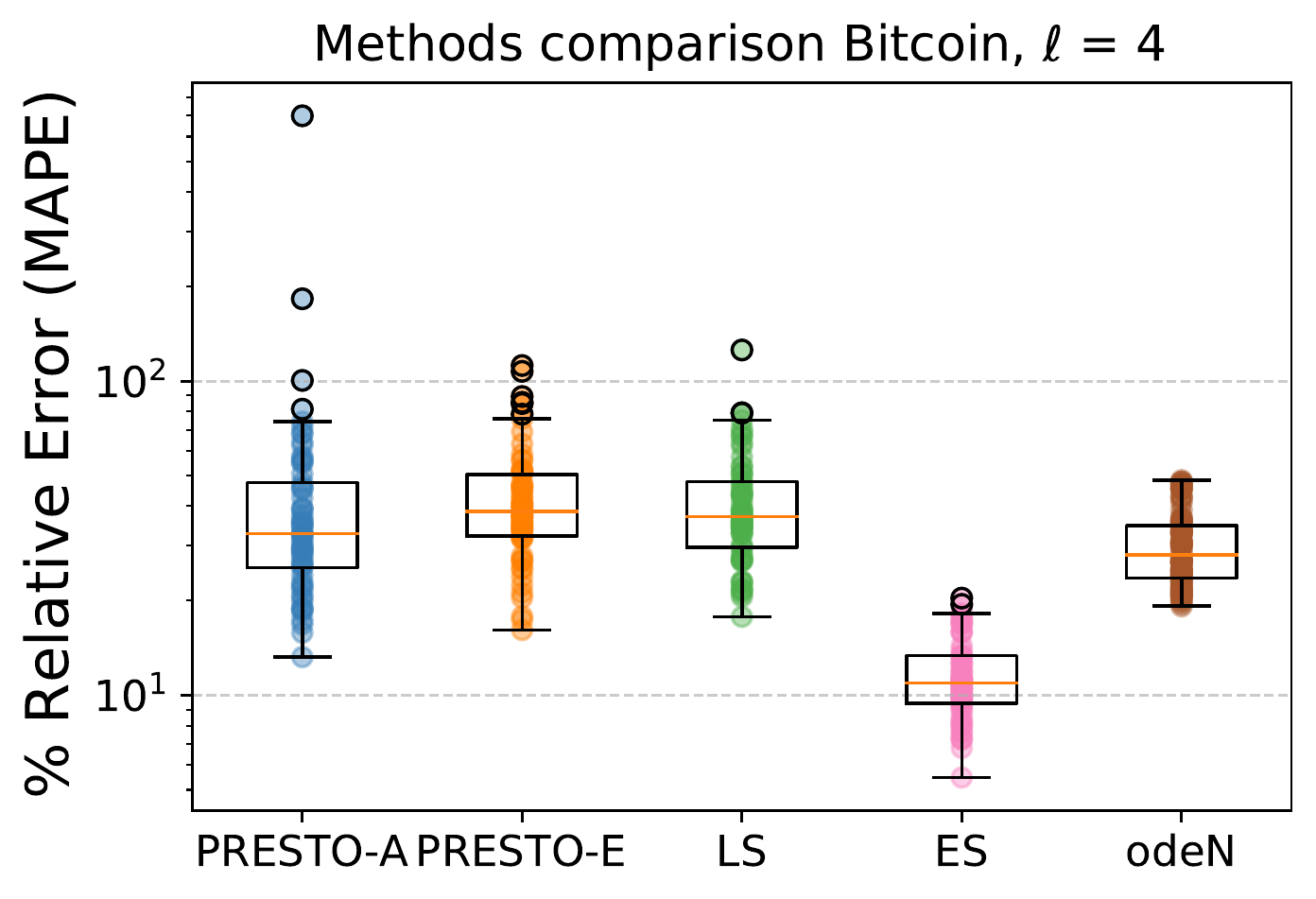} &
            \includegraphics[width=.45\linewidth]{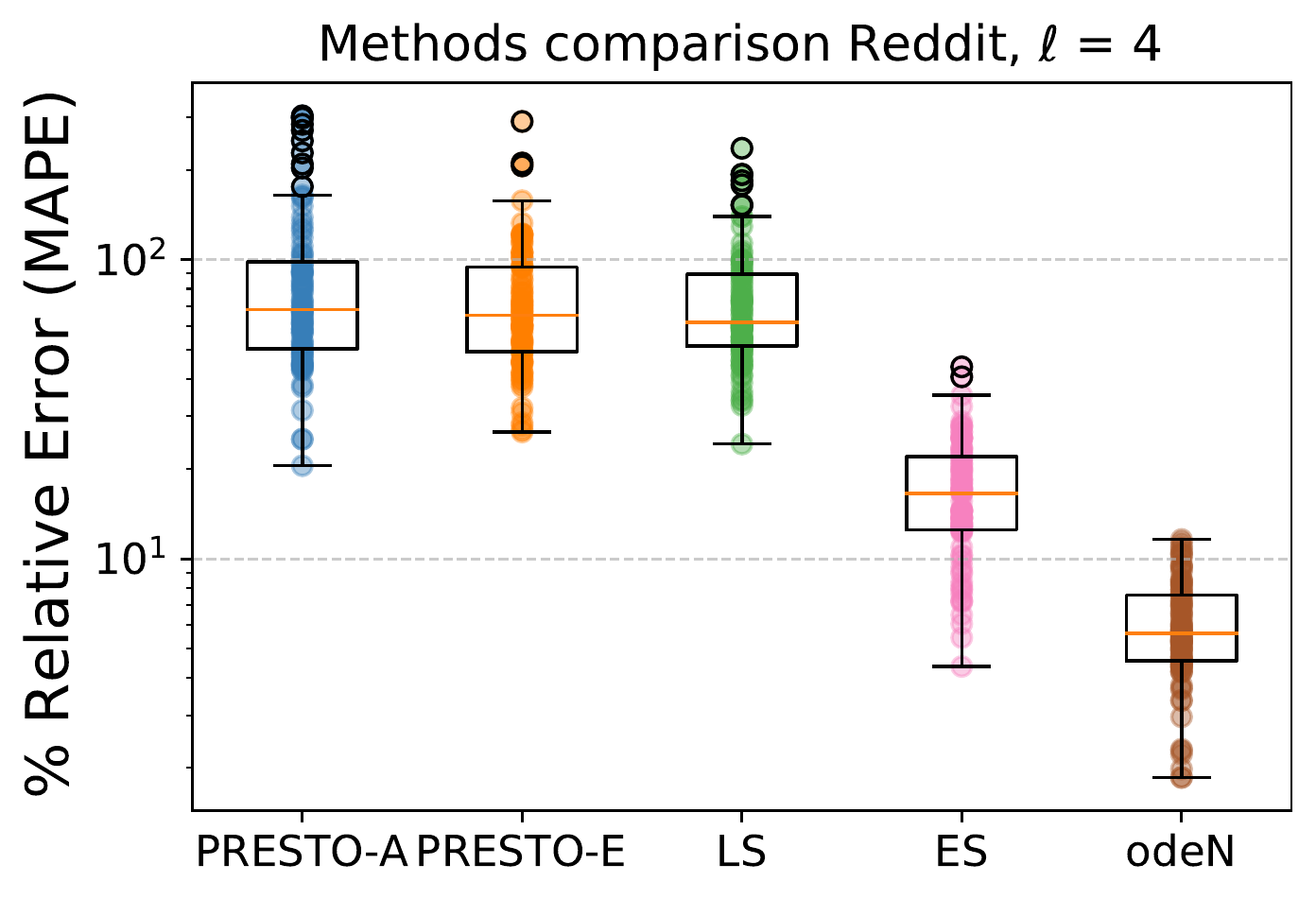} 
        \end{tabular}
        \label{fig:BI_RE_Triangles}
    }\\
    \subfloat[]{
        \begin{tabular}{cc}
            \includegraphics[width=.45\linewidth]{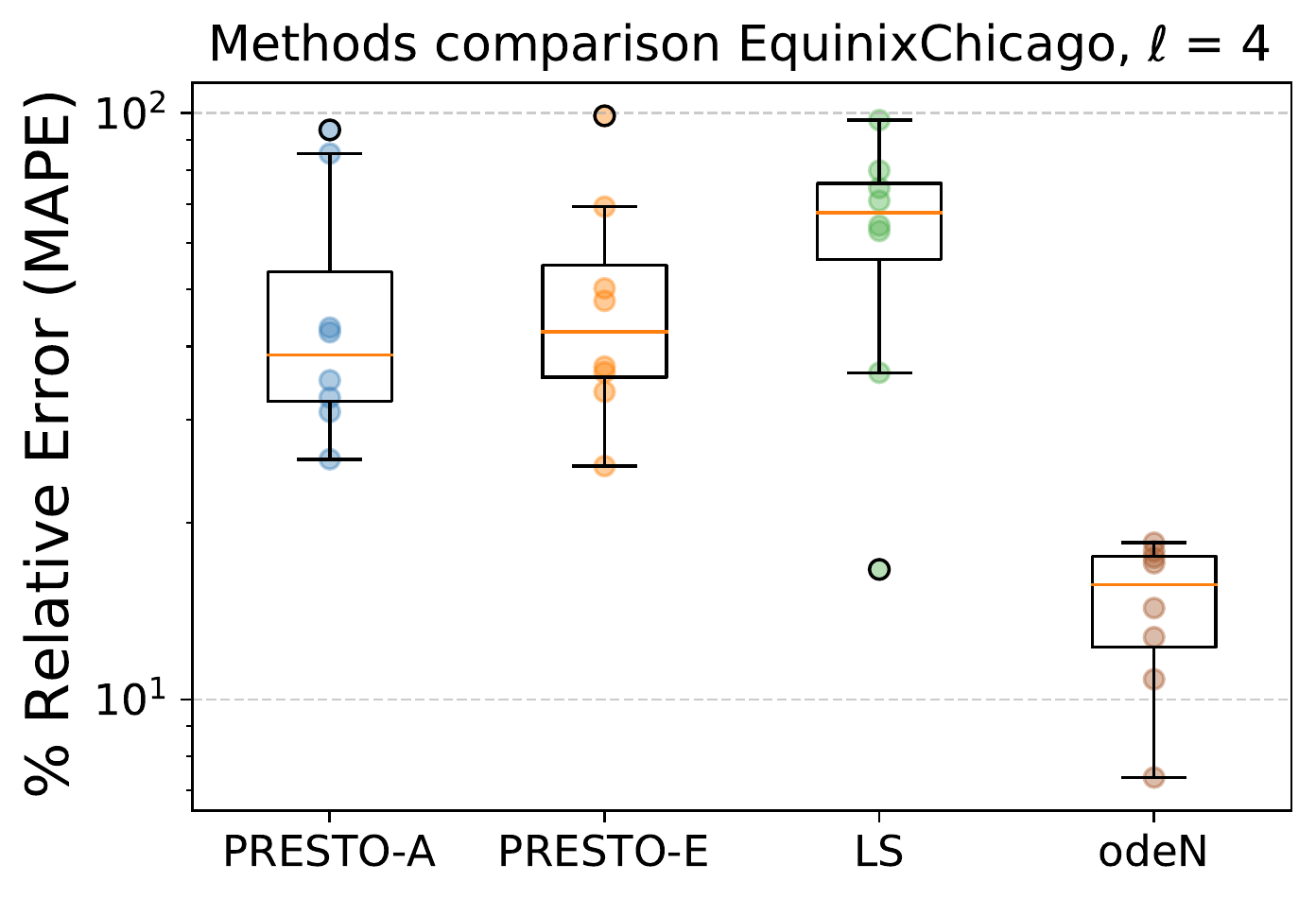} &
            \includegraphics[width=.45\linewidth]{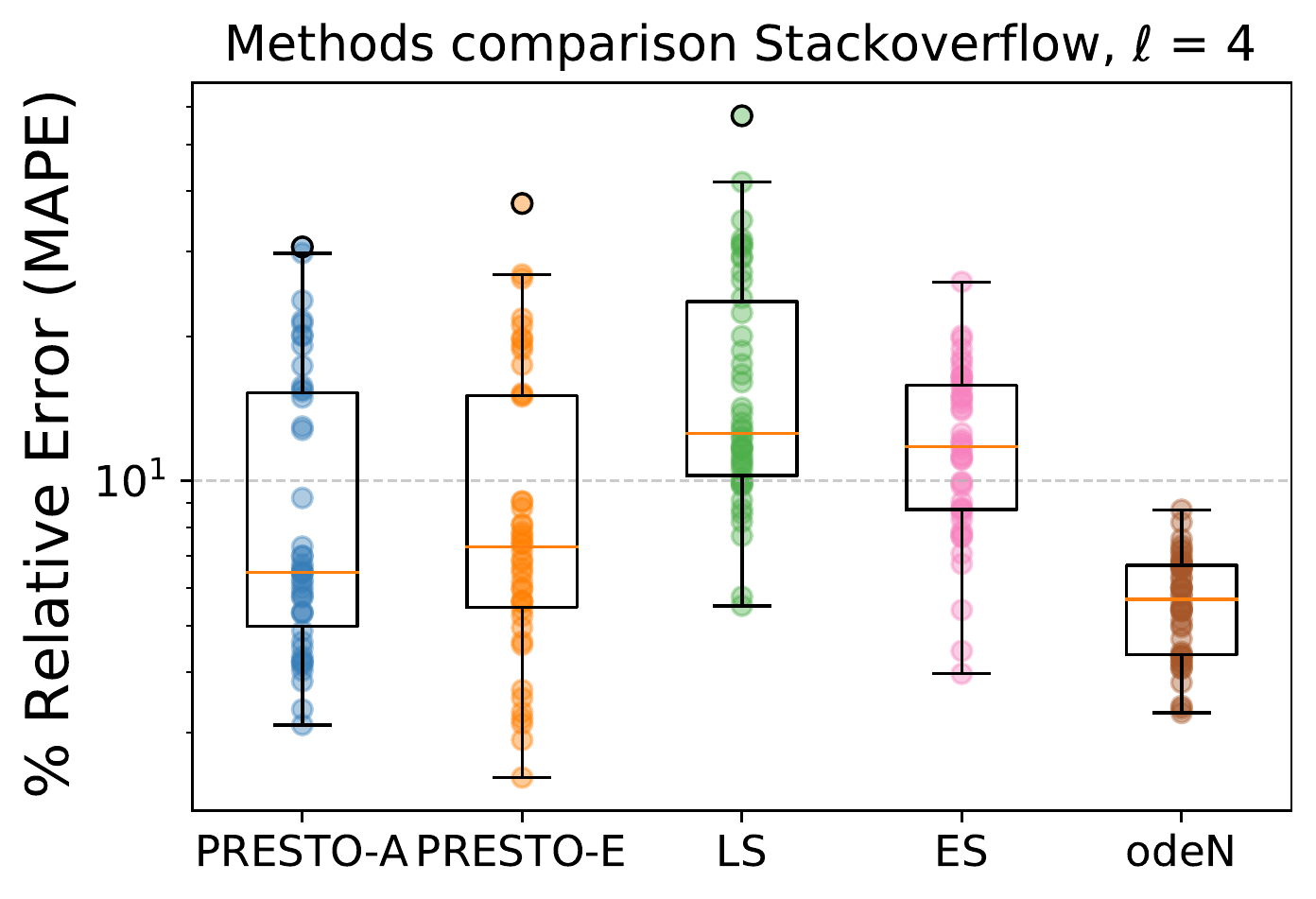} 
        \end{tabular}
        \label{fig:EC_SO_Estimates}
    }
    \caption{Approximation error on different datasets. (\ref{fig:SO_Triangles}): SO dataset, $H$ is a triangle, for $\ell=4$ (left) and $\ell=5$ (right). (\ref{fig:BI_RE_Triangles}): $H$ is a triangle, $\ell=4$, BI dataset  (left) and RE dataset (right). (\ref{fig:EC_SO_Estimates}): EC dataset, $H$ is an edge, $\ell=4$ (left); SO dataset, $H$ is a square, $\ell=4$.}
    \label{fig:approxBoxPlots}
\end{figure}

The results on the SO dataset are shown in Figure \ref{fig:SO_Triangles}. \algname\ provides much sharper estimates than state-of-the-art sampling techniques for single motif estimations on motifs $M_1,\dots,M_{|\mathcal{M}(H,\ell)|}$: the relative error on $\ell=4$-edge triangles is bounded by 5\%, and for $\ell=5$-edge triangles (where ${|\mathcal{M}(H,\ell)|}=800$) the relative error is bounded by 12\% while state-of-the-art algorithms report much less accurate estimates, with twice the relative error of \algname, on each configuration. We report the running times to obtain such estimates in Table~\ref{tab:triangleEstimatesTimes}.  Interestingly, \algname\ is more than 3$\times$ faster with $\ell=4$ than any sampling algorithm and 1.7$\times$ faster with $\ell=5$. 
For the other datasets, since extracting all the exact counts for $\ell >4$ is extremely time consuming, requiring up to months of computation, 
we will not discuss the approximation qualities for $\ell=5$ (since we do not have the exact counts to evaluate them).

On dataset BI (Figure \ref{fig:BI_RE_Triangles} left) \algname\ provides more concentrated estimates for the ${|\mathcal{M}(H,\ell)|} = 96$ triangles than other algorithms but \texttt{ES}, which also has a smaller running time than \algname. This may be related to the static graph structure of BI, which has some very high-degree nodes (see Table \ref{tab:datasets}). Therefore \algname\ may sample edges with very high degree nodes, introducing an over counting in its estimates.
Nonetheless, for higher values of $\ell$ this issue is amortized over the growing number of motifs ${|\mathcal{M}(H,\ell)|}$.

On dataset RE (Figure \ref{fig:BI_RE_Triangles} right) the estimates by \algname\ are all within 13\% of relative error and improve significantly over state-of-the-art sampling algorithms, up to one order of magnitude of precision. Such estimates were notably obtained with significantly smaller running time than state-of-the-art sampling algorithms, improving up to 2$\times$ the running time of \texttt{ES} and 1.4$\times$ over \texttt{PRESTO} (as reported in Table \ref{tab:triangleEstimatesTimes}).

Finally, on the EC datasets, which is a bipartite temporal network with more than 2 billion edges we evaluated the approximation qualities with $H$ being an edge and $\ell=4$ (for which ${|\mathcal{M}(H,\ell)|}=8$), such motifs have fundamental importance in the analysis of temporal networks since they can be seen as building blocks \cite{Holme2019Temporal,Zhao2010Communication}. We report the results on such motifs in Figure \ref{fig:EC_SO_Estimates} (left) (\texttt{ES} is not shown since it did not terminate with the allowed memory budget). 
The estimates of \algname\ are well concentrated and within 20\% of relative error, while other sampling approaches provide approximations with a relative error up to 90\% or more. Moreover, \algname's results were obtained with a speedup of at least 2$\times$ over all the other sampling algorithms, rendering the approximations task feasible in a small amount of time on very large temporal networks.

To illustrate the enormous advantage of \algname\ over existing state of the art exact and approximation algorithms, we compared the various algorithms on dataset SO when $H$ is set to be a square and $\ell=4$, for which ${|\mathcal{M}(H,\ell)|}=48$. As \cite{Wang2020Efficient} observed, among the 4-edge square motifs there are 16 motifs that do not grow as a single component (i.e., their orderings start with $\langle(1,2)(3,4)\cdots\rangle$). Estimating the counts of such motifs is particularly hard for most of the current state-of-the-art sampling algorithms since they generate a large number of partial matchings, while such aspect does not impact \algname. The results are shown in Figure \ref{fig:EC_SO_Estimates} (right). \algname\ provides tight approximations under 9\% of relative error for \emph{all} four-edge square motifs, while other sampling algorithms fail to provide sharp estimates for some of the motifs. Surprisingly, as shown in Table \ref{tab:triangleEstimatesTimes}, to obtain such estimates \algname\ required less than 1.3 hours of computation while the exact computation of the counts required more than two weeks, and \algname\  it is at least 3$\times$ times faster than all algorithms, and it is 5.4$\times$ times faster than \texttt{ES}.

\begin{table}[t]
    \caption{Running times (in seconds) to obtain the results in Figure \ref{fig:approxBoxPlots} (results are showed following the order in Figure \ref{fig:approxBoxPlots}). Under $H$ we report the topolology of $H$ used: \texttt{T} for triangles, \texttt{E} for edges, and \texttt{S} for squares. \enquote{-} denotes not applicable, while \enquote{\ding{55}} denotes out of RAM.}
    \scalebox{0.9}{
        \begin{tabular}{cccccccc}
            \toprule
            Dataset & $\ell$ & $H$& \texttt{PR-A} & \texttt{PR-E} & \texttt{LS}  & \texttt{ES} & \algname \\
            \midrule
            SO & 4& \texttt{T} & 533.4 & 537.7 & 555.5 & 567.2 & \textbf{174.4}\\
            SO & 5& \texttt{T} & 4405 & 4408 & 4390 & - & \textbf{2515}\\
            BI & 4 & \texttt{T}& 2048.6 & 2065.2 & 2754.6 &\textbf{1602.9} &1948.9\\
            RE& 4 & \texttt{T} &  9787.1 & 10165.8 & 14289.7 & 13172.3 & \textbf{6814.9}\\
            EC& 4 & \texttt{E} & 2581.5 & 3014.9 & 2981.9 & \ding{55} & \textbf{1234.3}\\
            SO & 4 & \texttt{S} & 15613.7 & 16718.7 & 14344.6 & 26118.3 & \textbf{4517.9}\\
            \bottomrule
        \end{tabular}
    }
    \label{tab:triangleEstimatesTimes}
\end{table}

Overall, these results show that our algorithm \algname\ achieves much more precise estimates within a significant smaller running time than state of the art sampling algorithms when estimating the counts $C_{M_1},\dots, C_{M_{|\mathcal{M}(H,\ell)|}}$ for different values of $\ell$ and different topologies of the  target template $H$ (see Problem~\ref{problem:TGMC} in Section~\ref{sec:prelims}).

\subsection{Parallel Implementation}\label{subsec:ParallelImpl}
\begin{figure}[t]
    \centering
    \begin{tabular}{cc}
        \includegraphics[width=.45\linewidth]{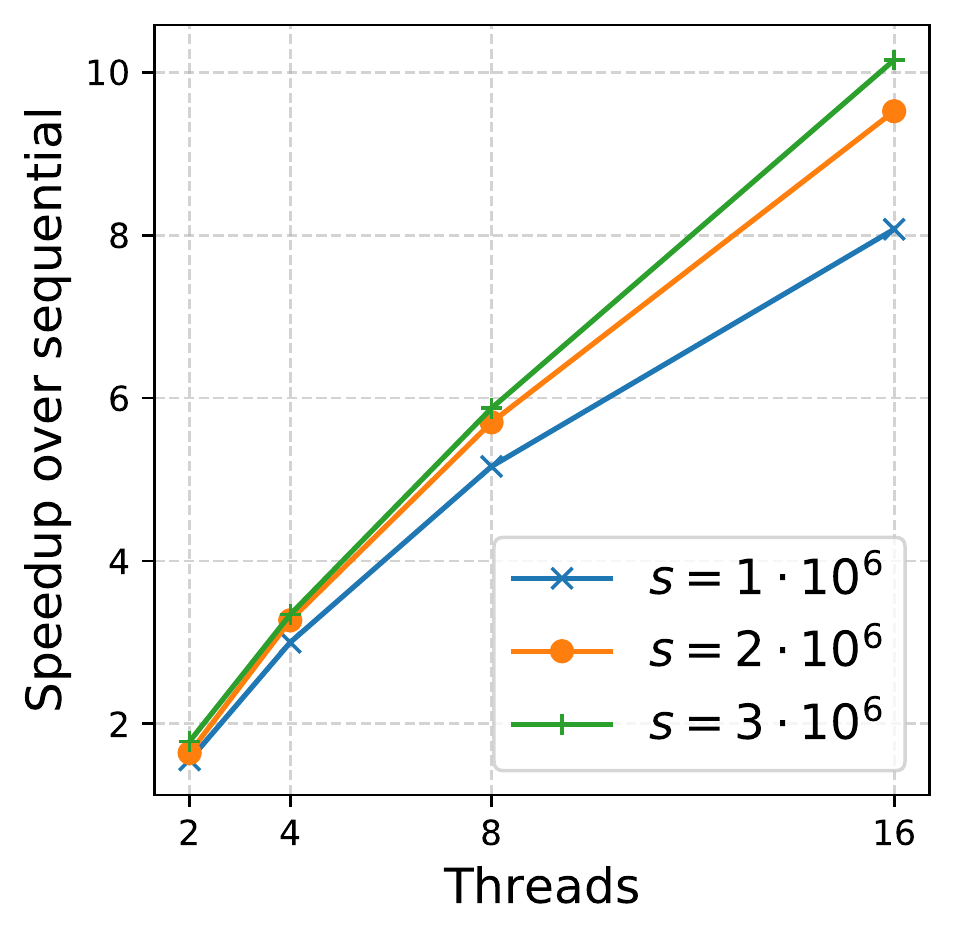} &
        \includegraphics[width=.45\linewidth]{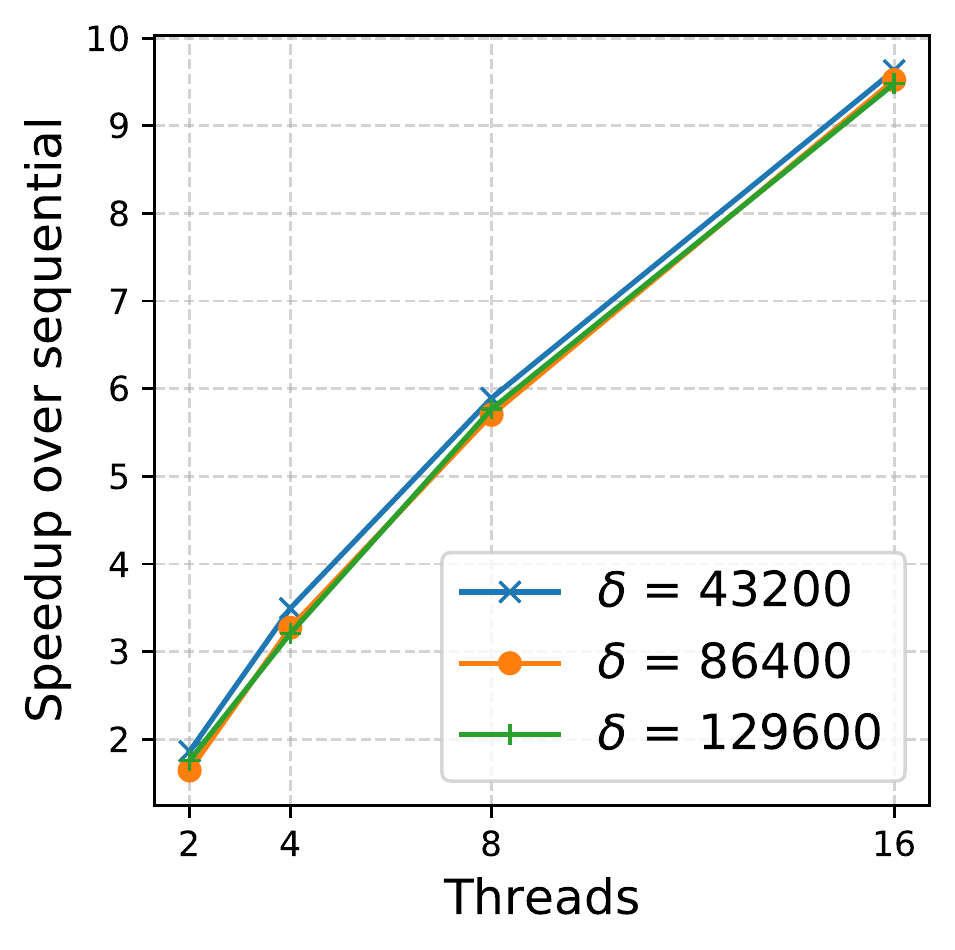}
    \end{tabular}
    \caption{Speed-up of \algname's parallel implementation. (Left): Varying $s$ and fixed $\delta$; (Right) Varying $\delta$ and fixed $s$.}
    \label{fig:parallelImpl}
\end{figure}
In this section we briefly describe the advantages of a simple parallel implementation of Algorithm \ref{alg:samplebystatic}. As discussed in Section \ref{subsec:generalAlg} the \texttt{for} cycle (from line \ref{alglinemain:forS}) can be trivially parallelized, therefore we implemented such strategy through a \emph{thread pooling} design pattern.

We describe the results obtained with $H$ set to be a triangle, $\ell=4$, and on the dataset SO; similar results are observed for other datasets. We tested the speedup achieved with $\omega \in\{2,4,8,16\}$ threads over the sequential implementation. Let $T_\omega$ the average running time with $\omega$ threads over ten execution of \algname\ with fixed parameters, with $T_1$ being the average time for running the algorithm sequentially.  
We report the value of  $T_1 / T_\omega, \omega \in \{2,4,8,16\}$, i.e., the speedup over the sequential implementation. Fig.\ \ref{fig:parallelImpl} (Left)
shows the speedup across different values of the sample size $s$, with $\delta=86400$. We observe an almost linear speedup up to 4 threads and then a slightly worse performance, especially for small sample sizes, that may be related to the time needed to process each sample. Fig.\ \ref{fig:parallelImpl} (Right) shows how the speedup changes for $s=2\cdot 10^6$ and different values of $\delta$. We note that our algorithm \algname\ seems not to be impacted by the value of $\delta$, and always attaining similar performances. 
Interestingly, as captured by our analysis in Section \ref{subsubsec:computComplex}, the algorithm does not reach a fully linear speedup since we did not parallelized the computation of the sampling probabilities $p(e), e\in E_T$. As a remark, our parallel implementation is not optimized, and more advanced parallel strategies may substantially increase its speedup.

\subsection{A Case Study}\label{subsec:caseStudy}
In this section we illustrate how counting multiple motifs, corresponding to the same target template $H$, with \algname\ can be used to extract useful insights from a temporal network. We consider a real-world activity network from Facebook~\cite{Viswanath2009}. In such network, each node represents a user and a temporal edge $(u,v,t)$ indicates that user $u$ posted on $v$'s wall at time $t$ (see the original publication~\cite{Viswanath2009} for more details). The network contains information collected from September 2006 to January 2009. After removing self-loops, the network has $n$=45.7K nodes, $m$=826K temporal edges, and $|E_T|$=179K static (undirected) edges. We will fist show how analyzing the motif counts obtained with \algname\ provides complementary insights to those in~\cite{Viswanath2009}, that relied on mostly static analyses. We then conclude by discussing how the counts of the network evolve by varying only the parameter $\ell$ (i.e., fixing $H, \delta$), showing that such counts surprisingly differ with different values of such parameter.

In the original paper~\cite{Viswanath2009}, the authors partitioned the Facebook network in nine different snapshots (obtaining nine projected static networks), with each snapshot spanning 90 days of interactions in the network. The authors observed that consecutive snapshots have small resemblance, i.e., on average only 45\% of the edges are preserved through consecutive snapshots. The authors also observed that despite this difference all the snapshots have similar, almost invariant, structural properties in terms of their clustering coefficient, average degree distribution, and others. We used \algname\ (with $\varepsilon=1, \eta=0.1$) to compare the temporal networks associated to the snapshots by computing the counts of the 8 temporal motifs in $\mathcal{M}(H,\ell=3)$ with $H$ being a triangle and $\delta = 86400$ = 1 day. On each snapshot, after extracting the motif counts, we computed for each motif $M$ its normalized count on the snapshot as $C_M/\sum_{i=1}^8C_{M_i}$. The results are reported in Fig. (\ref{subfig:countSnapshots}) (see Appendix \ref{app:CaseStudy} for a visual representation of the motifs). Interestingly, even if in~\cite{Viswanath2009} the authors highlight small resemblance through different snapshots, the counts of the motifs are stable across the different snapshots, especially by looking at the first three and the last two snapshots. Surprisingly on snapshots 6 and 7, which correspond to the period of observation of mid-2008, we observe that there is a significant variation in the motif counts w.r.t. the previous months. This is the period where the authors of~\cite{Viswanath2009} observed a change in Facebook's interface (that led to a drop in the growth of the network) that seems to be correlated to the variation on the motif counts. Even more surprisingly, this aspect is not captured by a static analysis of the snapshots as performed in~\cite{Viswanath2009}. Thus, our  temporal motifs analysis through \algname\ is able to capture a variation in the growth of the network that the static analysis cannot highlight. (We discuss how the motifs and their counts can be used to characterize the activity on the network in Appendix \ref{app:CaseStudy}).

We then analyzed how the different motif counts of the  whole network change by varying the parameter $\ell$.
We fixed $H$ a triangle and run \algname\ with $\varepsilon=1, \eta=0.1, \delta=86400$. The results are shown in Figure (\ref{subfig:countsFB}). We observe that the counts of $M_1, \dots, M_{|\mathcal{M}(H,\ell)|}$ vary significantly by increasing $\ell$. For $\ell=3$ almost all the motifs have the same counts, while for larger $\ell$ there are some motifs with very high counts (i.e., overrepresented) and some other motifs that are underrepresented. Overall the highest counts range from $10^4$ to $10^6$ from $\ell=3$ up to $\ell=6$. To understand if these counts increase only by chance, we performed a widely used statistical test (e..g,~\cite{Gauvin2018Randomized, Kovanen2013Gender}) by computing the $Z$-scores of the different motif counts under the following null model \cite{Milo2004Superfamilies}. We generated 500 random networks by the timeline shuffling random model
\cite{Gauvin2018Randomized}, which redistributes all the timestamps by fixing the \emph{directed} projected static network. For each motif $M_i, i=1,\dots,{|\mathcal{M}(H,\ell)|}$ we computed a $Z$-score that is defined as follows: let $C_{M_i}$ be the count of the motif in the original network and let $C_{M_i}^1, \dots, C_{M_i}^{500}$ be its counts on the $j$-th random network $j\in\{1,\dots,500\}$. The $Z$-score is computed as, $Z_{M_i} = (C_{M_i} - \sum_{j=1}^{500} C_{M_i}^j / 500 ) / \text{std}(C_{M_i}^1, \dots, C_{M_i}^{500})$ where std$(\cdot)$ denotes the standard deviation. The results are in Fig.~(\ref{subfig:ZscoresFB}), and they show that the counts in Fig.~(\ref{subfig:countsFB}) are very significant and not due to random fluctuations (higher $Z$-scores indicate that such motif counts are significantly more frequent in $T$ than in the networks permutated randomly). Interestingly, the $Z$-scores in Figure (\ref{subfig:ZscoresFB}) follow a similar law to the counts in Figure~(\ref{subfig:countsFB}), with the highest $Z$-scores increasing significantly every time $\ell$ increases. Notably the highest $Z$-scores of motifs with $\ell=6$ are  more than 3 orders of magnitude larger than the $Z$-scores of motifs with $\ell=3$. (We discuss some of the significant motifs in Appendix \ref{app:CaseStudy}).

\begin{figure}[t]
    \centering
    \subfloat[]{
        \includegraphics[width=0.8\linewidth]{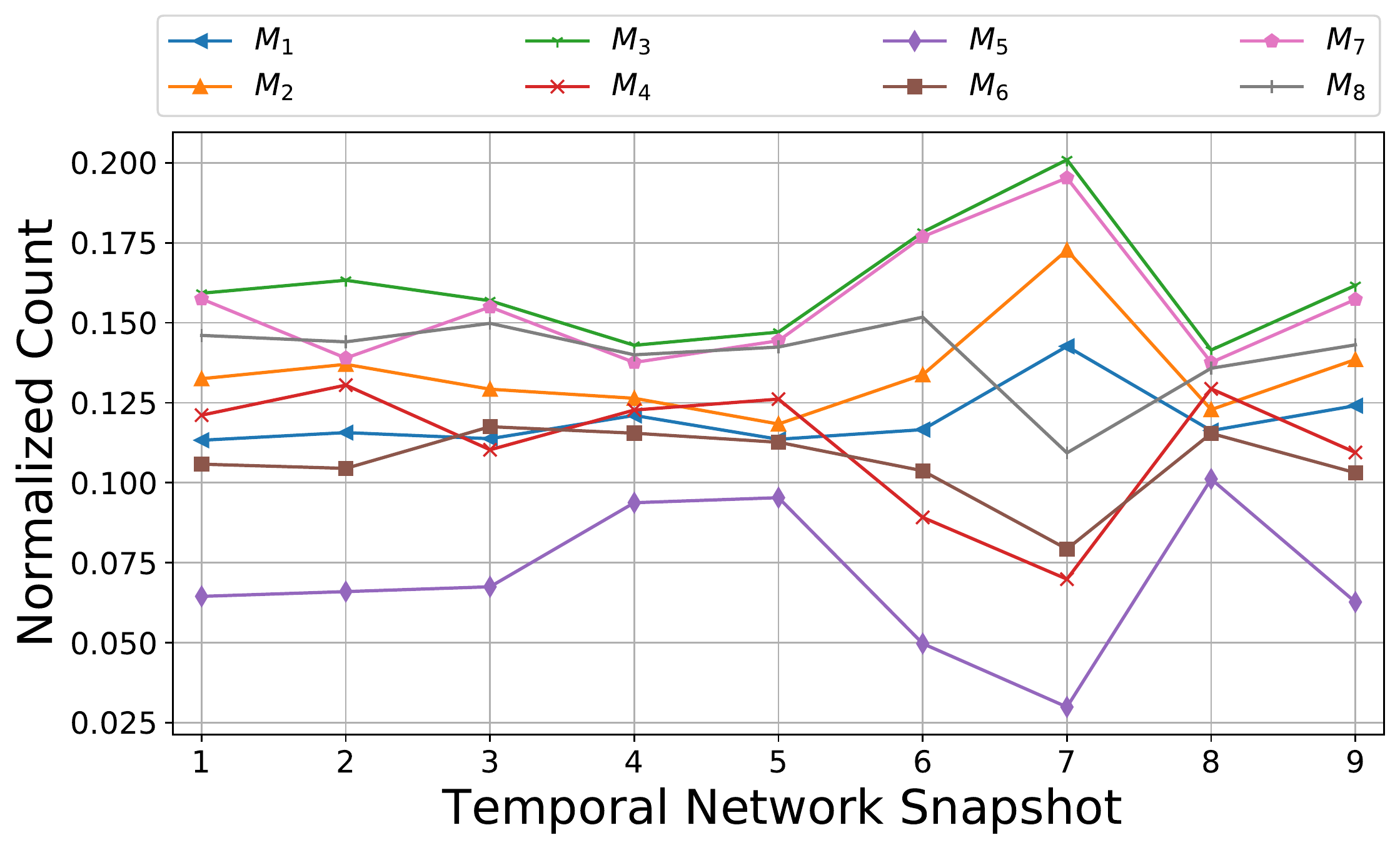}
        \label{subfig:countSnapshots}
    }
    \\
    \subfloat[]{
        \includegraphics[width=.45\linewidth]{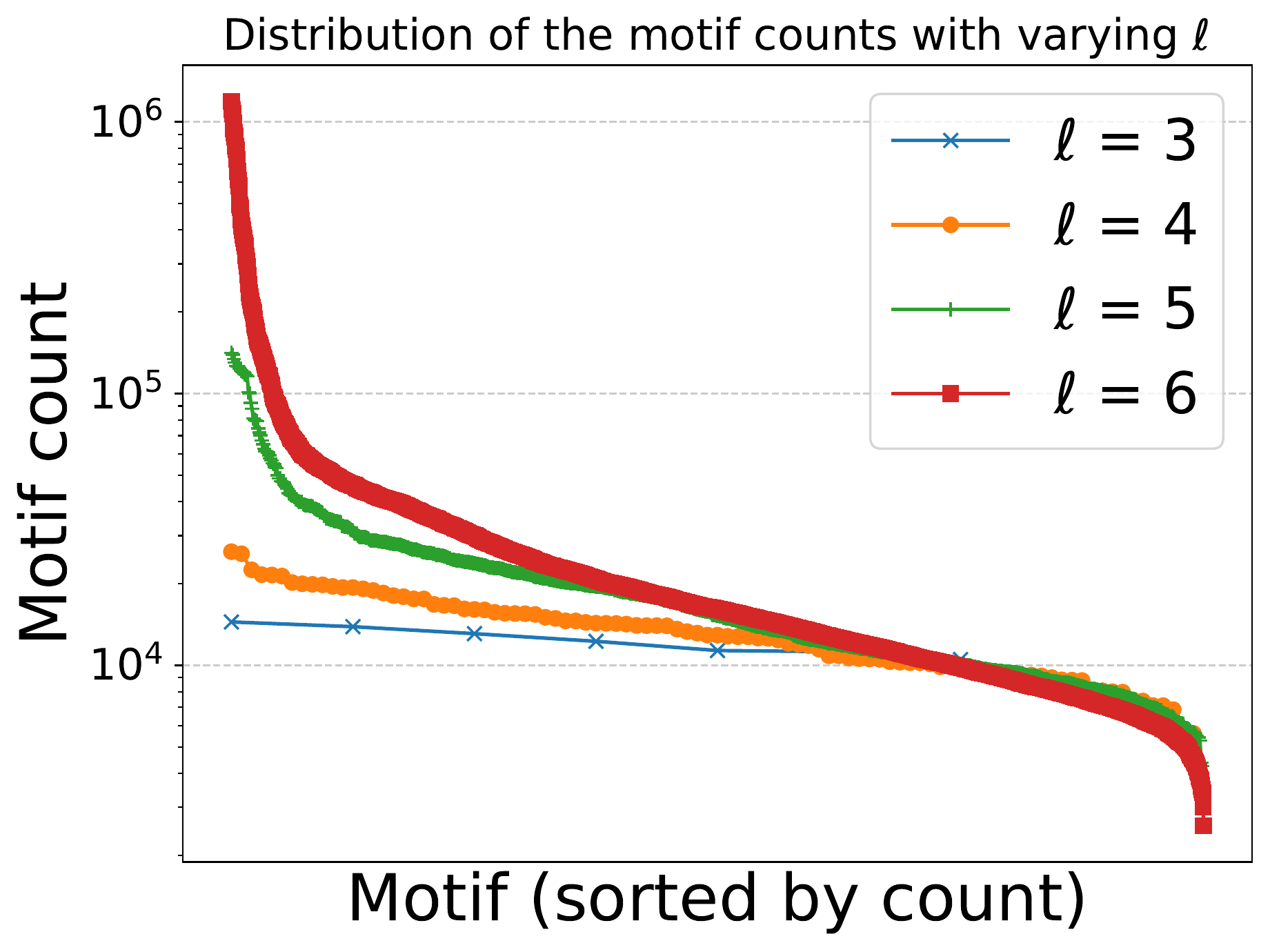}
        \label{subfig:countsFB}
    }
    \subfloat[]{
        \includegraphics[width=.47\linewidth]{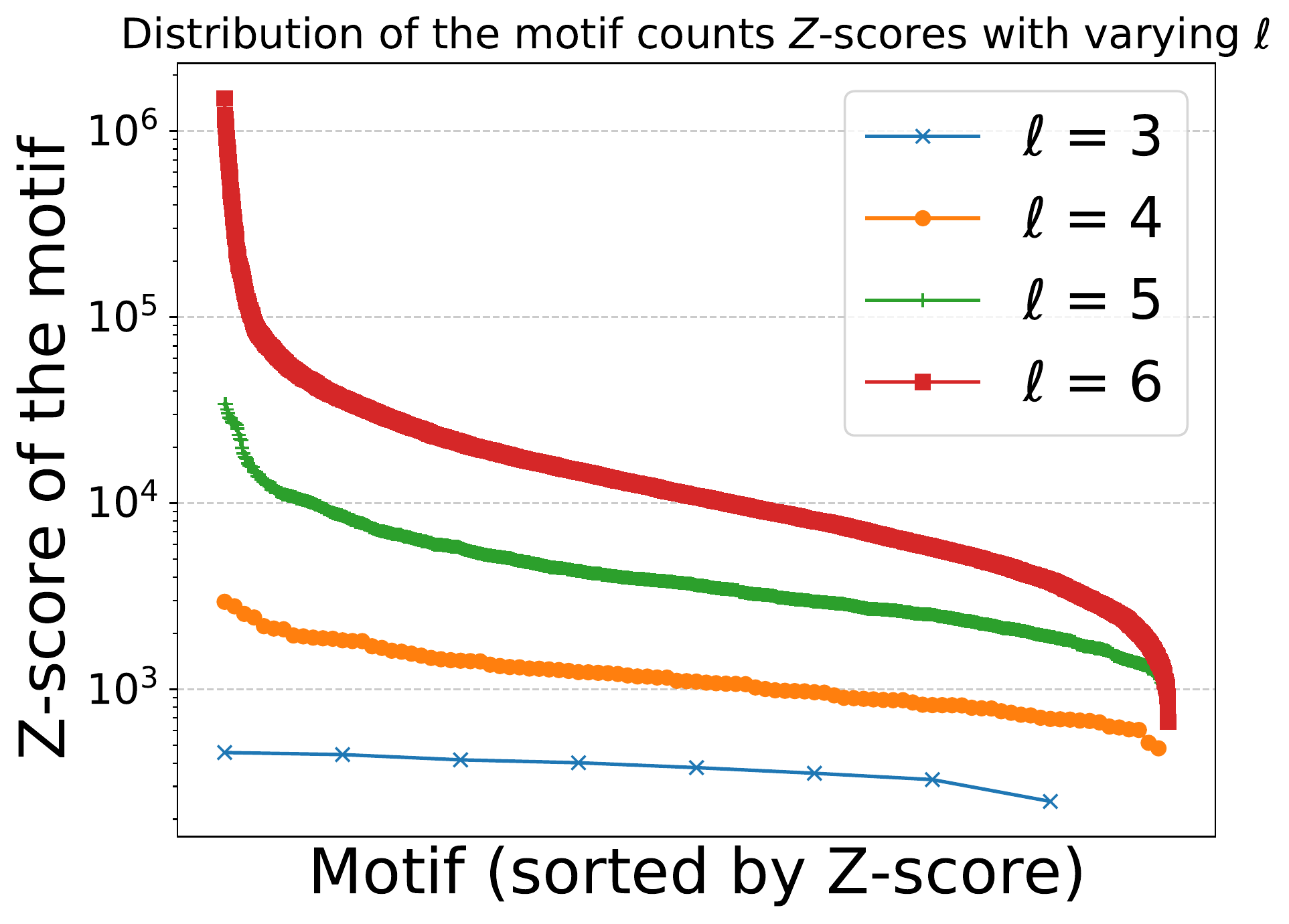}
        \label{subfig:ZscoresFB}
    }\\
    \caption{(\ref{subfig:countSnapshots}): Counts of the motifs in $\mathcal{M}(H,3)$ with $H$ a triangle on each temporal network corresponding to one snapshot in \cite{Viswanath2009}. (\ref{subfig:countsFB}): Counts on the full Facebook network with varying $\ell$. (\ref{subfig:ZscoresFB}): $Z$-scores of the motif counts with varying $\ell$.}
\end{figure}

\section{Conclusions}
In this work we introduced \algname, our algorithm to obtain rigorous, high-quality, probabilistic approximations of the counts of multiple motifs with the same static topology in large temporal networks. Our experimental evaluation shows that \algname\ allows to analyze several motifs in large networks in a fraction of the time required by state-of-the-art approaches. We believe that our algorithm \algname\ will be of practical interest in the analysis of temporal networks, complementing many of the existing tools and helping in understanding complex networked systems and their patterns.


There are several interesting directions for future research, including devising better edge probability distributions for \algname\ and choosing such distribution based on the characteristics of the dataset, since different datasets can have very different temporal edges distributions (e.g., with skewed behaviours \cite{Sarpe2021PRESTO}) and, thus, there may not exist a unique distribution that  is effective for all temporal networks.
Another direction of future research is the derivation of improved bounds for the number of samples required by \algname, using for example statistical learning theory concepts, such as pseudodimensions or Rademacher averages.

\begin{acks}
    This work was supported, in part, by MIUR of Italy, under PRIN Project n. 20174LF3T8 AHeAD, and grant L. 232 (Dipartimenti di Eccellenza), and by the U. of Padova project \enquote{SID 2020: RATED-X}.
\end{acks}

\newpage
\bibliographystyle{ACM-Reference-Format}
\bibliography{biblio}

\newpage
\appendix 
\section{Notation}
The notation used throughout this work is summarized in Table~\ref{tab:notationTable}.
\begin{table}[t]
    \centering
    \caption{Notation table.}
    \label{tab:notationTable}
    \begin{tabularx}{\columnwidth}{cl}
        \toprule
        Symbol & Description\\
        \midrule
        $T=(V,E)$ & Temporal network\\
        $n,m$ & Number of nodes and temporal edges of $T$\\
        $G_T$ & Undirected projected static network of $T$\\
        $M_i, i\in [1,{|\mathcal{M}(H,\ell)|}]$ & Motifs in $\mathcal{M}(H,\ell)$\\
        $k$ & Nodes in the motifs\\
        $\ell$ & Edges of the motifs\\
        $\delta$ & Duration limit of $\delta$-instances\\
        $M=(\mathcal{K}, \sigma)$ & Motif as pair (multigraph, ordering)\\
        $\mathcal{U}(M, \delta)$ & Set of $\delta$-instances of $M$ from $T$\\
        $C_M$ & Number of $\delta$-instances of $M$ in $T$\\
        $G_u[M]$ & Undirected graph associated to $\mathcal{K}$\\ 
        $\mathcal{M}(H,\ell)$ & \shortstack[l]{Set of \emph{distinct} motifs with $\ell$ edges s.t.\\ it holds $G_u[M_i]\simeq H\, \forall M_i\in \mathcal{M}(H,\ell)$.}\\
        $H$ & Static undirected target template\\
        $V_H, E_H$ & Set of nodes and edges of the target $H$\\
        $C_{M_i}(e)$ & Number of $\delta$-instances containing $e\in G_T$\\
        $s$ & Number of samples collected by \algname\\
        $X_e$ & Indicator variable denoting if $e\in G_T$ is sampled\\
        $p_e, p(e)$ & Probability of sampling edge $e\in G_T$\\
        $X_{M_i}^j$ & Estimate of motif $M_i$ obtained at \algname's $j$-th step\\
        $C_{M_i}'$ & Final \algname's estimate of $C_{M_i}$\\
        $\varepsilon, \eta$& Quality and confidence parameters\\ 
        $\omega$ & Number of threads in \algname\ parallel\\ 
        \bottomrule
    \end{tabularx}
\end{table}


\section{\algname's Subroutines} \label{app:subroutines}
\subsection{\texttt{FastUpdate} and its Subroutines}\label{subsec:FastUpdate} \sloppy{We now discuss the \texttt{FastUpdate} routine that is called in line \ref{alglinemain:fastupdatecall} of Algorithm \ref{alg:samplebystatic} to keep $C_{estimates}$ updated. The \texttt{FastUpdate} subroutine is shown in Algorithm~\ref{alg:fastcount}. $C_{estimates}$ maintains the weighted counts of the motif sequences identified, therefore to keep it updated we first count the $\delta$-instances of $M_i, i=1,\dots {|\mathcal{M}(H,\ell)|}$ within the sampled temporal network i.e.\ $S$, and then rescale each count opportunely. Such routine will feature two main aspects, i) an efficient adaptation of the algorithm by Paranjape et al.~\cite{paranjape2017motifs} and ii) an efficient encoding of the various sequences representing the motifs occurrences within integers that will allow for fast operations (comparisons to distinguish between different motifs and fast updates to the data structures).}

We now discuss how \texttt{FastUpdate} counts all the $\delta$-instances in $S$. First observe that we already know that $G_S \simeq H$, and that $S$ can be rewritten as $S=(((x_1,y_1), t_1), \dots, ((x_\ell,y_\ell), t_\ell) $. We first compute the set $E_{unique}= \{(x,y) : ((x,y), t) \in S\}$ and we assign to each edge in $E_{unique}$ a unique identifier (lines \ref{alglinefc:forEunique}-\ref{alglinefc:mapideedge}). Then we run an efficient implementation of the algorithm by Paranjape et al.\ \cite{paranjape2017motifs} that computes through dynamic programming the counts of all the subsequences of edges $(x,y)$ s.t.\ $(x,y,t)\in S$ having length $\ell$ and occurring within $\delta$-time (lines \ref{alglinefc:paranstartcycle}-\ref{alglinefc:increment}). 
In Algorithm \ref{alg:paranfast} we show our implementation of the subroutines needed to execute lines \ref{alglinefc:paranstartcycle}-\ref{alglinefc:increment} (see the original paper \cite{paranjape2017motifs} for full details and correctness). 
Intuitively, lines \ref{alglinefc:paranstartcycle}-\ref{alglinefc:increment} of Algorithm \ref{alg:fastcount} scan the input sequence $S$ linearly, maintaining in memory information about the edges within $\delta$ time from the processed one. 
Through such scan the algorithm updates $Map_{counts}$ to keep the counts of the sequences having at most $\ell$ edges over the set $E_{unique}$. 
Starting the  cycle in line \ref{alglinefc:cycleS}, $Map_{counts}$ contains the counts of all the $\ell$ subsequences of edges from $S$ over the set $E_{unique}$. 
We highlight that we assign to each static edge of $S$ an ID of $b$ bits. This allows us to encode each sequence up to $j=1,\dots,\ell$ edges, occurring within $\delta$ time, in an integer using $j\cdot b$ bits through bitwise operations (\enquote{$<<$} denotes right shift and \enquote{$|$} denotes bitwise or) 
to allow for fast updates to $Map_{counts}$. 

 \begin{algorithm}[t]
     \DontPrintSemicolon
     \SetKwComment{Comment}{$\triangleright$\ }{}
     \LinesNumbered 
     \KwIn{$\delta, S
         , C_{estimates}, p(e_R),H$}
     $E_{unique} \gets \{(x,y) : (x,y,t) \in S \} $\label{alglinefc:uniqueRdges}\;
     $Map_{id} \gets \{\}$, $E_{rev} \gets []$, $Map_{counts} \gets \{\}$, $start \gets 1$ \label{alglinefc:datastructMaps}\;
     \texttt{id} $\gets$ 0 \label{alglinefc:idEdge}\;
     \ForEach{$e \in E_{unique}$\label{alglinefc:forEunique}}
     {
         $E_{rev}[\texttt{id}] \gets e$ \label{alglinefc:reversearray}, 
         $Map_{id}\{e\} \gets$ \texttt{id++} \label{alglinefc:mapideedge}\;
     }
     \ForEach{$(x,y,t) \in S$ \label{alglinefc:paranstartcycle}} 
     {
         \While{$t - t_{start} > \delta $}
         {
             \texttt{Decrement}$( Map_{id}[(x_{start},y_{start})], Map_{counts})$\;
             $start \gets start +1$\;
         }
         \texttt{Increment}$(Map_{id}[(x,y)], Map_{counts})$ \label{alglinefc:increment}
     }
     \ForEach{$\text{key $\bar{k}$ of length }\ell \in Map_{counts}.keys$ \label{alglinefc:cycleS}}
     {
         $M' \gets $ \texttt{ReconstructMotif}($\bar{k}, E_{rev}$) \label{alglinefc:reconstructmotif}\;
         
         \If{$G_{u}[M'] \simeq H $ \label{alglinefc:conditionsM}}
         {
             $M_i \gets \texttt{EncodeAndClassifyMotif}(M')$ \label{alglinefc:classifymotif}\;
             $X_{M_i} \gets Map_{counts}\{\bar{k}\}/(|E_H|p(e_R))$ \label{alglinefc:getCount}\;
             $ X_{M_i}' \gets C_{estimates}\{M_i\}$\;
             $C_{estimates}\{M_i\} \gets X_{M_i}' + X_{M_i}$\label{alglinefc:updatesatructure}
         }  
     }
     
     \caption{\texttt{FastUpdate}}\label{alg:fastcount}
 \end{algorithm}

To obtain the estimates of motifs $M_1, \dots, M_{|\mathcal{M}(H,\ell)|}$, for each $\ell$ sequence of edges identified we reconstruct the corresponding graph and thus the motif $M'$ that the sequences is an instance of in line \ref{alglinefc:reconstructmotif} (the multigraph is given by the edges ID's while the ordering of the edges is given by the sequence itself). 
We then check if $G_{u}[M']$ is isomorphic to $H$ 
(constraint (1) from Problem \ref{problem:TGMC}). If so 
we encode the motif in a sequence of $2 b \ell$ bits that allows us to classify such motif (line \ref{alglinefc:classifymotif}) in order to distinguish between distinct motifs (recall we want $M_i \ncong_\tau M_j, i\neq j$). The encoding is computed as follows: given $M'= \langle (x_1,y_1),\dots, (x_\ell, y_\ell) \rangle$ we assign to each node an incremental ID according to its first appearance in $M'$ and we obtain the final encoding as $\langle\text{\texttt{ID}}(x_1)\text{\texttt{ID}}(y_1)\dots \text{\texttt{ID}}(x_\ell)\text{\texttt{ID}}(y_\ell)\rangle$. It is easily seen that two motifs $M_1,M_2$ share the same encoding iff it holds $M_1\cong_\tau M_2$ as desired, given that the motifs are \emph{directed} and the definition of distinct motifs accounts for the ordering in which edges appear. We provide an example below.

\begin{example}
    Let us consider $M_1, M_2,$ and $M_3$ from Figure \ref{fig:congruentMotifs}.  Consider $\sigma_1 = \langle(y,x), (y,z),(x,z)\rangle$, then by assigning an incremental ID to each node according to its first appearance in $\sigma_1$ we get $\text{\texttt{ID}}(y)=1, \text{\texttt{ID}}(x)=2,\text{\texttt{ID}}(z)=3$ so the final encoding of $M_1$ is $\langle121323\rangle$. Following a similar procedure the encoding of $M_2$ is $\langle121323\rangle$, while the encoding $M_3$ is $\langle121332\rangle$. The encodings of $M_1$ and $M_2$ coincide while differing from the one of $M_3$ as desired.
\end{example}
After this step we update the global data structure $C_{estimates}$ by summing to each motif's estimate, its count in $S$ divided by $|E_H| p(e_R)$ where $p(e_R)$ is the probability of edge $e_R$ of being sampled (lines \ref{alglinefc:getCount}-\ref{alglinefc:updatesatructure}), which we prove in Section \ref{sec:approx_analysis} to be the correct weighting schema to output an unbiased estimate.

\begin{algorithm}[t]
    \DontPrintSemicolon
    \SetKwFunction{algo}{Increment}\SetKwFunction{proc}{Decrement}
    \SetKwProg{myalg}{Function}{}{}
    \myalg{\algo{\texttt{id}, $Map_{counts}$}}{
        \nl \ForEach{$\text{$\bar{k}$} \in \texttt{\emph{SortByDecLength}}(Map_{counts}.keys)$}
        {
            \nl \If{$\bar{k}.length < \ell$}
            {
                \nl $\kappa \gets ({\bar{k}} << b) | \texttt{id}$\;
                \nl $Map_{counts}[\kappa] \gets Map_{counts}[\kappa] + Map_{counts}[\bar{k}]$\;
            }
            
        }
        \nl $Map_{counts}[\texttt{id}]\gets Map_{counts}[\texttt{id}] + 1$
    }{}
    \SetKwProg{myproc}{Function}{}{}
    \myproc{\proc{\texttt{id}, $Map_{counts}$}}{
        \nl $Map_{counts}[\texttt{id}]\gets Map_{counts}[\texttt{id}] - 1$ \;
        \nl \ForEach{$\text{$\bar{k}$} \in \texttt{\emph{SortByIncLength}}(Map_{counts}.keys)$}
        {
            \nl \If{$ \bar{k}.length < \ell - 1$}
            {
                \nl $\kappa \gets (\texttt{id} << (\bar{k}.length \cdot  b) )| \bar{k}$\;
                \nl $Map_{counts}[\kappa] \gets Map_{counts}[\kappa] - Map_{counts}[\bar{k}]$
            }
            
        }
    }
    \caption{Subroutines of \texttt{FastUpdate} }
    \label{alg:paranfast}
\end{algorithm}

    

\subsection{Exact Subgraph Enumeration}\label{subsubsec:exactRoutines}
In this section we briefly discuss the algorithms for subgraph enumeration that can be adapted to our Algorithm \ref{alg:samplebystatic} (in line \ref{alglinemain:exactsubgriso}). Unfortunately we cannot easily use the algorithms for extracting $k$-node motifs mentioned in Section \ref{sec:previousWorks} \emph{as is}, since they do not provide the local enumeration step required by \algname.

In fact, the problem most related to the exact enumeration we require is the \emph{labelled query graph matching problem}. In such setting one is provided a labelled query graph $H=(V_H,E_H,L_H)$, and a labelled graph $G=(V,E,L)$ (where labels can be colors for example, see \cite{Lee2012}), $L$ may be defined both on edges or vertices. The problem requires to find all the subgraphs $h'\subseteq G$ isomorphic to $H$, which could be either induced or not but must preserve the labelling properties (i.e., if $(x,y)\in E$ is mapped to $(x',y')\in H$ then $(L(x),L(y)) = (L_H(x'), L_H(y'))$). To explain how we take advantage of the algorithms developed for the problem above we need to introduce the following definitions (adapted from \cite{Pashanasangi2019}).
\begin{definition}\label{defn:vertex_orbits}
    Let $H=(V_H,E_H)$ be an undirected graph, an \emph{automorphism} is a bijection $\pi:V_H \mapsto V_H$ such that $(x,y)\in E_H$ iff $(\pi(x), \pi(y)) \in E_H$.
\end{definition}
\begin{definition}\label{defn:edge_automorphism}
    Let $H=(V_H,E_H)$ be an undirected graph, we say that two edges $e=(x,y), e'=(x',y') \in E_H$ belong to the same edge-orbit iff there exists an automorphism that maps $e$ on $e'$. 
\end{definition}
In order to adapt the algorithms for the labelled query graph matching problem we proceed in the following way: 1) colour the nodes of $G_T$ with a fixed colour (say red) 2) Once sampled $e_R \in G_T$, colour its endpoint nodes with a different colour (say blue), call the map from the last two points $L_{G_T}$; 3) compute the different edge-orbits of the pattern $H$ (by enumerating the automorphisms of $H$) and for each edge-orbit choose an edge, colour its endpoint nodes with the same colour assigned to $e_R$, and keep the colour on the other edges the same as $G_T$, call this map $L_H$; 4) run an algorithm for the labelled query graph matching problem with graph $G_T=(V_T,E_T,L_{G_T})$ and pattern $H=(V_H,E_H,L_H)$ 5) the desired subgraphs ($\mathcal{H}$) are the union over the different edge-orbits enumeration steps. 

\section{Implementation Details}\label{app:implementation}
In this section we provide additional implementation details, complementing the description of Section \ref{subsec:expsetup}.

In our implementation, we used two main structures: first, an adjacency list\footnote{We used the one provided by SNAP: \url{https://github.com/snap-stanford/snap}, more efficient implementations can be also adopted improving the global running times.}, that allows to query for an edge between $u,v\in V$ in $O(\log(\min({d_u,d_v})))$. Second, we used a hashmap to store for each static directed edge the timestamps of the temporal edges that map on that edge, leading to $O(1)$ complexity of querying for the timestamps of a static edge in $G_T$. The initialization of such structures is done in $O(1)$ per each processed temporal edge while loading the dataset, by knowing the number of nodes $n$. Many state of the art algorithms exist for the local enumeration of motifs (e.g., \cite{Sun2020, Ren2015, Han2013}), we provide in our code a general algorithm based on the algorithm VF2++ \cite{Juettner2018}. However, instead of using the general procedure described in Section~\ref{subsubsec:exactRoutines}, in our test we relied on a simple algorithm that locally enumerates the subgraphs containing an edge $e=\{x,y\}$ isomorphic to $H$:  for triangles the algorithm runs in $O(\min({d_x,d_y})\log(n))$, while when $H$ is a square the algorithm runs in $O(\min({d_x,d_y})d_{max}\log(n))$, with $d_{max}$ the maximum degree of a node in $G_T$.

\section{Proofs} \label{app:proofs}
In this section we provide the proofs not included in the main text. 

First we recall that $C_{M_i}(e)$ the number of $\delta$-instances of motif  $M_i, i=1,\dots,|\mathcal{M}(H,\ell)|$ from $T$ whose undirected projected static network contains edge $e\in G_T$, i.e., $C_{M_i}(e) = \sum_{h \subseteq G_T, h \simeq H: e \in h} |\mathcal{U}(h,M_i)|, e \in G_T$ where $\mathcal{U}(h,M_i)$ is the set of $\delta$-instances of motif $M_i$ whose static projected graph is $h\subseteq G_T$. Then based on the above it is simple to notice that the following formula holds for each motif $M_i,i=1,\dots,|\mathcal{M}(H,\ell)|$: $\sum_{e\in G_T} C_{M_i}(e) = |E_H|C_M$. This relation will be the key for proving the unbiasedness of the estimates provided by \algname, as we show next.
 
\begin{proof}[Proof of Lemma \ref{lemma:unbiased}]
    First let us consider the expectation of $X_{M_i}^j, i=1,\dots, |\mathcal{M}(H,\ell)|, j=1,\dots,s$:
    \[
    \begin{aligned}
    &\mathbb{E} \left[ \frac{1}{|E_H|} \sum_{e \in G_T}\frac{ C_{M_i}(e) X_e}{p_e} \right] = \frac{1}{|E_H|} \sum_{e \in G_T} \frac{C_{M_i}(e)\mathbb{E}[X_e]}{p_e} = C_{M_i}\\
    \end{aligned}
    \]
    where we used the linearity of expectation and the facts that $\mathbb{E}[X_e] = p_e, e\in G_T$, and $\sum_{e \in G_T} C_{M_i}(e)= |E_H| C_{M_i}$; thus $X_{M_i}^j, i=1,\dots, |\mathcal{M}(H,\ell)|,j=1,\dots,s$ are unbiased estimates of $C_{M_i}$, combining such result to $C_{M_i}'$ we obtain,
    \[
    \mathbb{E}[C_{M_i}'] = \mathbb{E} \left[ \frac{1}{s} \sum_{j=1}^s X_{M_i}^j \right] =  \frac{1}{s} \sum_{j=1}^s \mathbb{E}[X_{M_i}^j] =\frac{s C_{M_i}}{s}= C_{M_i}
    \]
    by the linearity of expectation.
\end{proof}

\begin{proof}[Proof of Lemma \ref{lemma:variancebound}]
    We need to bound the variance of the estimate $C_{M_i}'$, first we rewrite the estimator
    \[
    C_{M_i}' =  \frac{1}{s} \sum_{j=1}^s \frac{1}{|E_H|}\sum_{e \in G_T} C_{M_i}(e)\frac{X_e}{p_e} = 
    \frac{1}{s} \sum_{j=1}^s X_{M_i}^j
    \]
    Since the $s$ variables $X_{M_i}^j, j\in[1,s]$ are independent (edges are drawn independently at each iteration of the outer \texttt{for} loop in Algorithm \ref{alg:samplebystatic}), it holds
    $\text{var}(C_{M_i}') = \text{var}(\frac{1}{s} \sum_{j=1}^s X_{M_i}) = \frac{1}{s}\text{var}(X_{M_i}) $ we thus only need to compute the variance of the variable $X_{M_i}$. Let us recall $\text{var}(X_{M_i})=\mathbb{E}[X_{M_i}^2]- \mathbb{E}[X_{M_i}]^2 = \mathbb{E}[X_{M_i}^2]- C_{M_i}^2$ by the previous lemma. We will now bound $\mathbb{E}[X_{M_i}^2]$.
    \[
    \begin{aligned}
    &\mathbb{E}[X_{M_i}^2] = 
    \mathbb{E}\left[ \frac{1}{|E_H|^2} \sum_{e_1 \in G_T} \sum_{e_2 \in G_T} C_{M_i}(e_1)C_{M_i}(e_2)\frac{X_{e_1}X_{e_2}}{p_{e_1}p_{e_2}}\right] \\
    & = \frac{1}{|E_H|^2} \sum_{e_2 \in G_T} C_{M_i}^2(e_2)\frac{1}{p_{e_2}} 
    \leq \frac{1}{|E_H|^2} \sum_{e_2 \in G_T} C_{M_i}^2(e_2) \frac{m}{\alpha} = \\
    &=\frac{m}{\alpha |E_H|^2}\sum_{e_2 \in G_T}C_{M_i}^2(e_2) \overset{(1.)}{\leq} \frac{m}{\alpha |E_H|^2} |E_H| C_{M_i}^2 = \frac{mC_{M_i}^2}{\alpha|E_H|}
    \end{aligned}
    \]
    where we used the linearity of expectations, the fact that $\mathbb{E}[X_{e_1}X_{e_2}]= p_{e_1}$ only for $e_1=e_2$ otherwise is 0, a bound on the minimum probability $p_e$ where $p_e \leq \alpha/m, \forall e \in G_T$ for $\alpha$ defined as in Section \ref{sec:approx_analysis}. In $(1.)$ we used the fact that $C_{M_i}(e) = \lambda_e C_{M_i}, e \in G_T, \lambda_e \in [0,1]$, then $\sum_{e_2 \in G_T}C_{M_i}^2(e_2) = \sum_{e_2 \in G_T}\lambda_{e_2}^2 C_{M_i}^2 \leq C_{M_i}^2 \sum_{e_2 \in G_T}\lambda_{e_2} = |E_H| C_{M_i}^2$ since $\lambda_{e_2} \in [0,1]$ and further $\sum_{e \in G_T}\lambda_{e} = |E_H|$ by $\sum_{e \in G_T}\lambda_{e}C_{M_i} = |E_H|C_{M_i}$.
    
    Thus the variance of $X_{M_i}$ is bounded by: 
    \[
    \text{Var}(X_{M_i}) \leq \frac{mC_{M_i}^2}{\alpha|E_H|} - C_{M_i}^2 = C_{M_i}^2 \left(\frac{m}{\alpha|E_H|} -1\right)
    \]
    combining everything together we obtain that $
    \text{var}(C_{M_i}') \leq \frac{C_{M_i}^2}{s} $ $ \left(\frac{m}{\alpha|E_H|} -1\right)$, concluding the proof.
\end{proof}

\begin{proof}[Proof of Theorem \ref{theo:boundsamplesmultiple}]
    Let us fix $M_i, i\in[1,|\mathcal{M}(H,\ell)|]$ we first show a bound to the following probability $\mathbb{P}[|C_{M_i}' - C_{M_i}| \ge \varepsilon C_{M_i}]$.
    We want to derive such bound through the application of Bennett's inequality to the following summation: $\frac{1}{s} \sum_{j=1}^s X_{M_i}^j$, we already know that $\mathbb{E}[X_{M_i}^j] = C_{M_i}$ and $\mathbb{E}[(X_{M_i}^j - C_{M_i})^2] \leq C_{M_i}^2\left(\frac{m}{\alpha|E_H|} -1\right) = \hat{v}^2_j$ for $j=1,\dots,s$ it holds: 
    \[
    X_{M_i}^j= \frac{1}{|E_H|}\sum_{e \in G_T} C_{M_i}(e)\frac{X_e}{p_e} \le \frac{1}{|E_H|}\sum_{e \in G_T} C_{M_i}(e) \frac{m}{\alpha}= \frac{mC_{M_i}}{\alpha|E_H|}
    \]
    As argued by \cite{Sarpe2021PRESTO} Bennett's inequality holds even if we only have an upper bound on the variance of the estimates. Therefore let us compute the quantities to apply Bennett's bound (see \cite{Sarpe2021PRESTO} for the statement), clearly $B = C_{M_i}(\frac{m}{\alpha|E_H|}-1)$ combining what we already showed with the unbiasedness of $X_{M_i}^j$, moreover $v \le \hat{v}^2_j$ since the bound $\hat{v}^2_j$ is equal for each $j\in[1,s]$. Then, 
    \[
    \frac{\hat{v}^2_j}{B^2} = \frac{C_{M_i}^2\left(\frac{m}{\alpha|E_H|} -1\right)}{C_{M_i}^2(\frac{m}{\alpha|E_H|}-1)^2} = \frac{1}{(\frac{m}{\alpha|E_H|}-1)}
    \]
    also 
    \[
    \frac{tB}{\hat{v}^2_j} = \frac{\varepsilon C_{M_i} C_{M_i}(\frac{m}{\alpha|E_H|}-1)}{C_{M_i}^2\left(\frac{m}{\alpha|E_H|} -1\right)} = \varepsilon
    \]
    Combining everything together by Bennett's inequality we obtain,
    \begin{equation}\label{eq:boundBennett}
    \mathbb{P} \left( \left|\frac{1}{s}\sum_{j=1}^s X_{M_i}^j - C_{M_i} \right| \ge \varepsilon C_{M_i} \right) \le 2\exp \left( - \frac{s}{(\frac{m}{\alpha|E_H|}-1)} h(\varepsilon) \right)
    \end{equation}
    
    Now, let $A_i = ``|C_{M_i}' - C_{M_i}| \ge \varepsilon C_{M_i}", i=1,\dots,|\mathcal{M}(H,\ell)|$, namely $A_i$ is the event that the estimate of motif $M_i, i=1,\dots,|\mathcal{M}(H,\ell)|$ is distant more than $\varepsilon C_{M_i}$ from $C_{M_i}$. We already showed that that for an arbitrary $A_i$ inequality \eqref{eq:boundBennett} holds for $\mathbb{P}[A_i]$, so
    \[
    \begin{aligned}
    &\mathbb{P}\left(\bigcup_{i=1}^{|\mathcal{M}(H,\ell)|} A_i\right) \leq \sum_{i=1}^{|\mathcal{M}(H,\ell)|}\mathbb{P}[A_i] \leq\\
    &\leq{|\mathcal{M}(H,\ell)|} 2\exp \left( - \frac{s}{(\frac{m}{\alpha|E_H|}-1)} h(\varepsilon) \right) \leq \eta
    \end{aligned}
    \]
    combining the union bound and the choice of $s$ as in statement.
\end{proof}
\section{Case Study - Motif analysis}\label{app:CaseStudy}
\begin{figure}[h]
    \centering
        \includegraphics[width=\linewidth]{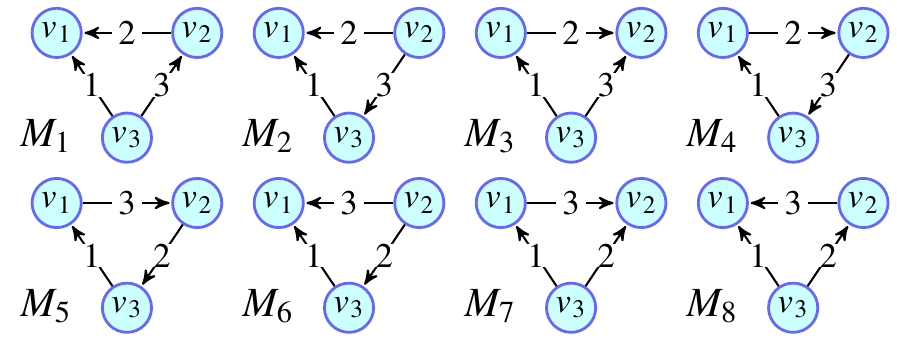}
        
    \caption{Graphical representation of the motifs in Figure (\ref{subfig:countSnapshots}).}
    \label{fig:triangleMap}
\end{figure}
\subsubsection*{Motifs on the Snapshots of the Facebook Network.} Thanks to our analysis of Section \ref{subsec:caseStudy} we are able to characterize the user behaviour on the Facebook network of wall posts by looking at different motifs (topology and their orderings) and their counts. We first show in Fig. \ref{fig:triangleMap} the motifs corresponding to the labels of Figure (\ref{subfig:countSnapshots}) in Section \ref{subsec:caseStudy}. Then, let $H=\{v_1,v_2,v_3\}$ be a triangle, the most frequent motifs (i.e., those with the highest normalized counts on each snapshot) seem to share a common pattern: a first node ($v_3$) after posting on $v_1$'s (or $v_2$'s) wall triggers $v_1$ (or $v_2$'s) to post on the remaining node's wall with $v_1$ posting also on such node's wall to close the triangle, as captured by motifs $M_3$, $M_7$ and $M_8$. Observe that by identifying the users that mostly act as $v_3$ in the occurrences of such frequent motifs one is able to identify, for example, the nodes more engaged in spreading most of the information over the Facebook network in a short period of time (recall that we set $\delta$ to one day). Not surprisingly motif $M_5$ is the less frequent one since its occurrences require node $v_2$ to post on $v_3$'wall before receiving the post from $v_2$ therefore without being \enquote{triggered} by such node, that received the post from $v_3$. Interestingly, without considering the orderings of occurrence among such patterns we will not be able to distinguish between the most frequent motifs and the least frequent ones since for example $M_4$ and $M_5$ have the same static \emph{directed} graph structure but they have very different counts on the different snapshots of the Facebook network.

\subsubsection*{Motifs with varying $\ell$ - Frequent vs Infrequent.}
\begin{figure}[h]
    \centering
    \includegraphics[width=\linewidth]{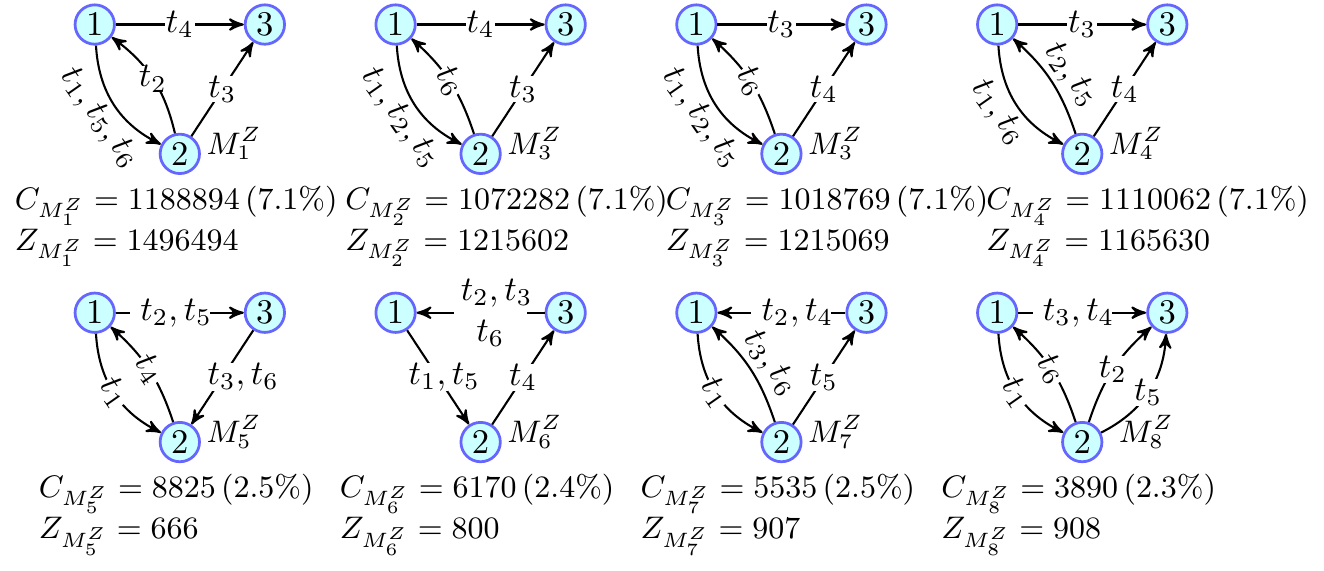}
    
    \caption{Graphical representation of the 4 motifs with highest (top) and lowest (bottom) $Z$-scores in Figure (\ref{subfig:ZscoresFB}) for $\ell=6$. For each motif we report the exact count (which we computed for such representation) and the relative error in the approximation obtained with \algname\ in brackets, we additionally report each $Z$-score of the motif as obtained from Section \ref{subsec:caseStudy} (i.e., by using only \algname).}
    \label{fig:toplowZscores}
\end{figure}
In this Section we briefly discuss the properties and show visually the motifs with highest and lowest $Z$-scores obtained in Section \ref{subsec:caseStudy} on the Facebook wall post network for $\ell=6$. The motifs are reported in Figure \ref{fig:toplowZscores}, where we report the 4-top motifs ranked by $Z$-score on the top and the 4-lowest motifs by $Z$-scores on the bottom. Note, that the top 4 motifs share a similar structure, both temporal and topological. Interestingly in the original paper \cite{Viswanath2009} the authors noted that there were very few pair of nodes that exchanged more than 5 messages (with median 2). The most frequent temporal motifs seem to involve a pair of highly active nodes (which exchanged many messages between them, i.e., more than 4) and another third node that is reached by such pair of nodes. We unfortunately do not have the original messages to understand better the information captured by such frequent motifs (since we do not have the original posts), but it is really surprising that the top 4 motifs all share similar properties especially in the orderings of their edges. Additionally, it seems that triangles involving nodes that are pairwise very active seem to be the rarest type of interaction as captured by the 4 motifs with lowest $Z$-score, reported in Figure \ref{fig:toplowZscores} bottom.
\end{document}